\documentclass[journal]{IEEEtran}
\usepackage{amsmath,graphicx}

\usepackage{amsmath,amssymb,epsfig,psfrag,cite,subfigure}
\include{macros}
\usepackage{graphicx}
\usepackage{epstopdf}

\usepackage{bm}

\usepackage{geometry}
\usepackage{tcolorbox}

\usepackage{soul}

\usepackage{amsthm}

\usepackage{nicefrac}

\usepackage{lipsum,multicol}

\usepackage{algorithm}
\usepackage{algpseudocode}

\algdef{SE}[SUBALG]{Indent}{EndIndent}{}{\algorithmicend\ }%
\algtext*{Indent}
\algtext*{EndIndent}

\usepackage{enumitem}
\usepackage{mwe}    
\usepackage{epsfig,psfrag}
\usepackage{subfigure}
\usepackage{color}
\usepackage{url}
\usepackage{mathtools,xparse}

\usepackage{gensymb}

\usepackage[multiple]{footmisc}

\makeatletter
\newcommand*{\addFileDependency}[1]{
  \typeout{(#1)}
  \@addtofilelist{#1}
  \IfFileExists{#1}{}{\typeout{No file #1.}}
}
\makeatother

\theoremstyle{remark}

\newtheoremstyle{mytheoremstyle} 
    {\topsep}                    
    {\topsep}                    
    {\upshape}                   
    {.5em}                           
    {\itshape}                   
    {.}                          
    {.5em}                       
    {}  
    
\theoremstyle{mytheoremstyle}

\newtheoremstyle{iremark}
  {\topsep}   
  {\topsep}   
  {\upshape}  
  {0.2in}       
  {\itshape}  
  {.}         
  {5pt plus 1pt minus 1pt} 
  {\thmname{#1}\thmnumber{ \itshape#2}\thmnote{ (#3)}} 

\theoremstyle{iremark}

\newtheorem{lemma}{Lemma}
\newtheorem{remark}{Remark}

\newtheorem*{proof}{Proof}

\DeclareFontFamily{U}{mathx}{\hyphenchar\font45}
\DeclareFontShape{U}{mathx}{m}{n}{
	<5> <6> <7> <8> <9> <10>
	<10.95> <12> <14.4> <17.28> <20.74> <24.88>
	mathx10
}{}
\DeclareSymbolFont{mathx}{U}{mathx}{m}{n}
\DeclareFontSubstitution{U}{mathx}{m}{n}

\DeclarePairedDelimiter\abs{\lvert}{\rvert}%
\DeclarePairedDelimiter\absbig{\Big\lvert}{\Big\rvert}%
\DeclarePairedDelimiter\absbigs{\big\lvert}{\big\rvert}%

\DeclarePairedDelimiter\floor{\Bigg\lfloor}{\Bigg\rfloor}
\DeclarePairedDelimiter\floorr{\lfloor}{\rfloor}

\makeatletter \renewcommand\d[1]{\ensuremath{%
		\;\mathrm{d}#1\@ifnextchar\d{\!}{}}}
\makeatother

\makeatletter
\newcommand*\rel@kern[1]{\kern#1\dimexpr\macc@kerna}
\newcommand*\widebar[1]{%
  \begingroup
  \def\mathaccent##1##2{%
    \rel@kern{0.8}%
    \overline{\rel@kern{-0.8}\macc@nucleus\rel@kern{0.2}}%
    \rel@kern{-0.2}%
  }%
  \macc@depth\@ne
  \let\math@bgroup\@empty \let\math@egroup\macc@set@skewchar
  \mathsurround\z@ \frozen@everymath{\mathgroup\macc@group\relax}%
  \macc@set@skewchar\relax
  \let\mathaccentV\macc@nested@a
  \macc@nested@a\relax111{#1}%
  \endgroup
}
\makeatother

\newcommand{\norm}[1]{\left\lVert#1\right\rVert}

\newcommand{\normbig}[1]{\Big\lVert#1\Big\rVert}

\newcommand{\thnew}[1]{ {#1^{\rm{th} } } }

\newcommand{\Eee}{\mathbb{E}}

\newcommand{\hdet}{ \underset{\mathcal{H}_0}{\overset{\mathcal{H}_1}{\gtrless}} }

\newcommand{\sigmabar}{\widebar{\sigma}}
\newcommand{\sigmatilde}{\widetilde{\sigma}}

\newcommand{\snr}{{\rm{SNR}}}
\newcommand{\pfa}{P_{\rm{fa}}}

\newcommand{\sigmax}{\sigma_x}

\newcommand{\Tsym}{ T_{\rm{sym}} }
\newcommand{\Tcp}{ T_{\rm{cp}} }

\newcommand{\FF}{ \mathbf{F} }

\newcommand{\deltaf}{ \Delta f }
\newcommand{\fc}{ f_c }
\newcommand{\Ntx}{ N_{\rm{T}} }
\newcommand{\Nrx}{ N_{\rm{R}} }
\newcommand{\ftx}{ \ff_{\rm{T}} }

\newcommand{\atx}{ \aaa_{\rm{T}} }
\newcommand{\arx}{ \aaa_{\rm{R}} }

\newcommand{\wwconj}{ \ww^{*} }
\newcommand{\wwhat}{ \widehat{\ww} }

\newcommand{\hhtilde}{ \widetilde{\hh} }

\newcommand{\sm}{ s_{m} }
\newcommand{\hhhat}{ \widehat{\hh} }
\newcommand{\hhtildehat}{ \widehat{\hhtilde} }

\newcommand{\llr}{\mathcal{L}}
\newcommand{\llrlog}{\mathcal{L}^{\rm{log}}}

\newcommand{\projrange}[1]{\boldsymbol{\Pi}_{#1}}
\newcommand{\projnull}[1]{\boldsymbol{\Pi}^{\perp}_{#1}}

\newcommand{\xnm}{ x_{n,m} }

\newcommand{\boldHhat}{ \widehat{\boldH} }

\newcommand{\wtheta}{\mathcal{W}_{\thetahat}}

\newcommand{\alphahat}{ \widehat{\alpha} }

\newcommand{\deltar}{ \Delta R }
\newcommand{\deltav}{ \Delta v }
\newcommand{\Rmax}{ R_{\rm{max}} }
\newcommand{\vmax}{ v_{\rm{max}} }

\newcommand{\numax}{ \nu_{\rm{max}} }
\newcommand{\numaxcfo}{ \nu^{\rm{cfo}}_{\rm{max}} }

\newcommand{\rect}[1]{ { \rm{rect} }\left(#1\right) }

\newcommand{\Smcal}{ \mathcal{S} }

\newcommand{\mub}{ \boldsymbol{\mu} }

\newcommand{\mtCN}{{\mathcal{CN}}}

\newcommand{\vecc}[1]{ {\rm{vec}}\left(#1\right)  }

\newcommand{\diag}[1]{ {\rm{diag}}\left(#1\right)  }

\newcommand{\Imatrix}{{ \boldsymbol{\mathrm{I}} }}

\newcommand{\aaa}{\mathbf{a}}
\newcommand{\cc}{ \mathbf{c} }
\newcommand{\bb}{ \mathbf{b} }

\newcommand{\nuhat}{{ \widehat{\nu} }}
\newcommand{\tauhat}{{ \widehat{\tau} }}
\newcommand{\thetahat}{{ \widehat{\theta} }}
\newcommand{\nuhatcfo}{{ \nuhat^{\rm{cfo}} }}

\newcommand{\boldzero}{{ {\boldsymbol{0}} }}

\newcommand{\Pmax}{{ P_{\rm{max}} }}

\newcommand{\ff}{\mathbf{f}}

\newcommand{\boldY}{ \mathbf{Y} }
\newcommand{\boldYfs}{ \mathbf{Y}^{\rm{FS}} }
\newcommand{\hhatfs}{ \boldHhat^{\rm{FS}} }

\newcommand{\boldX}{ \mathbf{X} }

\newcommand{\Xbbar}{ \widebar{\boldX} }
\newcommand{\Ybbar}{ \widebar{\boldY} }
\newcommand{\Zbbar}{ \widebar{\boldZ} }

\newcommand{\boldZ}{ \mathbf{Z} }

\newcommand{\boldQ}{ \mathbf{Q} }
\newcommand{\boldH}{ \mathbf{H} }

\newcommand{\boldD}{ \mathbf{D} }

\newcommand{\boldA}{ \mathbf{A} }
\newcommand{\boldR}{ \mathbf{R} }

\newcommand{\boldU}{ \mathbf{U} }

\newcommand{\boldLambda}{ \mathbf{\Lambda} }
\newcommand{\etatilde}{\widetilde{\eta}}

\newcommand{\yy}{ \mathbf{y} }

\newcommand{\xx}{ \mathbf{x} }
\newcommand{\hh}{ \mathbf{h} }
\newcommand{\ww}{ \mathbf{w} }
\newcommand{\zz}{ \mathbf{z} }

\newcommand{\DD}{ \mathbf{D} }

\newcommand{\boldPhi}{ \mathbf{\Phi} }

\newcommand{\transpose}[1]{ {#1}^{T} }

\newcommand{\complexset}[2]{ \mathbb{C}^{#1 \times #2}  }

\newcommand{\alphabar}{ \widebar{\alpha} }

\graphicspath{{./Figures/}}

\begin{document}

\title{
MIMO-OFDM Joint Radar-Communications: \\Is ICI Friend or Foe?}

\author{Musa Furkan Keskin, \textit{Member, IEEE}, Henk Wymeersch, \textit{Senior Member, IEEE}, and Visa Koivunen, \textit{Fellow, IEEE}\thanks{Musa Furkan Keskin and Henk Wymeersch are with the Department of Electrical Engineering, Chalmers University of Technology, SE 41296 Gothenburg, Sweden (e-mail: furkan@chalmers.se). Visa Koivunen is with the Department of Signal Processing and Acoustics, Aalto University, FI 00076 Aalto, Finland. This work is supported, in part, by Vinnova grant 2018-01929, MSCA-IF grant
888913 (OTFS-RADCOM) and the European Commission through the H2020 project Hexa-X (Grant Agreement no. 101015956).\\
This study extends our previous work in \cite{ICI_aware_ICASSP_2021} by providing additional contributions including GLRT-OMP based multi-target detector design and ICI exploitation capability.}}

\maketitle

\begin{abstract}
    Inter-carrier interference (ICI) poses a significant challenge for OFDM joint radar-communications (JRC) systems in high-mobility scenarios. In this paper, we propose a novel ICI-aware sensing algorithm for MIMO-OFDM JRC systems to detect the presence of multiple targets and estimate their delay-Doppler-angle parameters. First, leveraging the observation that spatial covariance matrix is independent of target delays and Dopplers, we perform angle estimation via the MUSIC algorithm. For each estimated angle, we next formulate the radar delay-Doppler estimation as a joint carrier frequency offset (CFO) and channel estimation problem via an APES (amplitude and phase estimation) spatial filtering approach by transforming the delay-Doppler parameterized radar channel into an unstructured form. To account for the presence of multiple targets at a given angle, we devise an iterative interference cancellation based orthogonal matching pursuit (OMP) procedure, where at each iteration the generalized likelihood ratio test (GLRT) detector is employed to form decision statistics, providing as by-products the maximum likelihood estimates (MLEs) of radar channels and CFOs. In the final step, target detection is performed in delay-Doppler domain using target-specific, ICI-decontaminated channel estimates over time and frequency, where CFO estimates are utilized to resolve Doppler ambiguities, thereby turning ICI from foe to friend. The proposed algorithm can further exploit the ICI effect to introduce an additional dimension (namely, CFO) for target resolvability, which enables resolving targets located at the same delay-Doppler-angle cell.
    Simulation results illustrate the ICI exploitation capability of the proposed approach and showcase its superior detection and estimation performance in high-mobility scenarios over conventional methods.

	\vspace{-0.1in}
\end{abstract}


\vspace{-0.1in}
\section{Introduction}\label{sec_intro}
\vspace{-0.05in}
With the explosive growth of spectrally co-existent radars and communication systems in 5G and beyond wireless networks, joint radar-communications (JRC) strategies have become popular in recent years \cite{jointRadCom_review_TCOM,SPM_JRC_2019,SPM_Zheng_2019,chiriyath2017radar,Eldar_SPM_JRC_2020,Canan_SPM_2020}. 
A promising approach to practical JRC deployment is to design \textit{dual-functional radar-communications} (DFRC) systems, which employ a single hardware that can simultaneously perform radar sensing and data transmission with a co-designed waveform \cite{DFRC_SPM_2019,DFRC_Waveform_Design,SPM_JRC_2019}. Orthogonal
frequency-division multiplexing (OFDM) has been widely investigated as a DFRC waveform due its wide availability in wireless communication systems and its potential to achieve high radar performance \cite{RadCom_Proc_IEEE_2011,General_Multicarrier_Radar_TSP_2016,ICI_OFDM_TSP_2020}. In the literature, estimator design for OFDM radar sensing has been studied in both single-antenna \cite{RadCom_Proc_IEEE_2011,Firat_OFDM_2012,ICI_OFDM_TSP_2020} and multiple-input multiple-output (MIMO) \cite{mmWave_JRC_TAES_2019,MIMO_OFDM_radar_TAES_2020} settings.

In high-mobility scenarios, such as millimeter-wave (mmWave) vehicular JRC systems \cite{SPM_JRC_2019}, Doppler-induced intercarrier interference (ICI) can significantly degrade the performance of OFDM from both radar and communications perspective \cite{OFDM_ICI_TVT_2017,ICI_OFDM_radar_2015,multiCFO_TWC_2018}. Due to lack of guard intervals in frequency domain, Doppler shifts produced by mobile targets can destroy the orthogonality of OFDM subcarriers at the receiver (i.e., the ICI effect) and reduce the dynamic range of radar due to increased side-lobe levels \cite{OFDM_ICI_TVT_2017}. To improve OFDM radar performance, various ICI mitigation approaches have been recently proposed \cite{Firat_OFDM_2012,ICI_OFDM_radar_2015,tian2017high,lellouch2014impact}. Considering a single-target scenario, a two-step Doppler estimation method is proposed in \cite{ICI_OFDM_radar_2015}, where transmit sequences with favorable correlation properties are employed to eliminate ICI. In a similar fashion, the studies in \cite{lellouch2014impact,tian2017high} design ICI compensation schemes for a single target relying on the condition that transmit symbols consist of phase codes with certain characteristics regarding cyclic shifts and auto-correlation. As a step further, the work in \cite{Firat_OFDM_2012} assumes arbitrary phase shift keying (PSK) symbols and proposes a pulse compression technique to compensate for ICI-induced phase rotations across OFDM subcarriers created by a single target, which can compensate for Doppler shifts at integer multiples of subcarrier spacing. In all of the above schemes, the common observation is that ICI can be harnessed to resolve Doppler ambiguity of a single target.

To tackle the more practical scenarios of multiple targets, recent works investigate range-velocity estimation techniques for multiple objects in the presence of ICI \cite{OFDM_ICI_TVT_2017,ICI_OFDM_TSP_2020,lowComp_ICI_2020}. Under the assumption that the number of targets is known a-priori, an alternating maximization approach that takes ICI into account is proposed in \cite{ICI_OFDM_TSP_2020} to reduce the complexity of high-dimensional maximum-likelihood (ML) search. In \cite{OFDM_ICI_TVT_2017}, the ICI effect is eliminated via an all-cell Doppler correction (ACDC) method, which performs Doppler compensation for arbitrary velocities with a precision on the level of Doppler resolution. Following a similar line of reasoning, the work in \cite{lowComp_ICI_2020} develops a sparsity based ICI removal and range-velocity estimation algorithm. For proper functioning of the methods in \cite{OFDM_ICI_TVT_2017,lowComp_ICI_2020} in terms of ICI compensation, a common requirement is that OFDM symbol matrix should be rank-one (i.e., same symbols repeated over time). 

The prior approaches to OFDM radar sensing in the presence of ICI \cite{ICI_OFDM_radar_2015,tian2017high,lellouch2014impact,OFDM_ICI_TVT_2017,ICI_OFDM_TSP_2020,lowComp_ICI_2020} suffer from several major drawbacks that can significantly limit their applicability in practical JRC scenarios. First, most of the existing ICI compensation methods focus exclusively on the radar functionality of OFDM and impose strict constraints on transmit data symbols, such as cyclic shift property \cite{tian2017high}, good correlation characteristics \cite{ICI_OFDM_radar_2015,tian2017high}, identical sequences across subcarriers \cite{lellouch2014impact} and rank-one frequency-time OFDM symbol matrix \cite{OFDM_ICI_TVT_2017,lowComp_ICI_2020}. Such constraints, however, would lead to a substantial loss in communications data rate (e.g., $M$ times reduction in data rate \cite{OFDM_ICI_TVT_2017,lowComp_ICI_2020}, where $M$ is the number of OFDM symbols) and essentially impede dual-functional operation, which is one of the core properties of OFDM. In addition, prior works either focus on a single-target scenario \cite{ICI_OFDM_radar_2015,tian2017high,lellouch2014impact}, which is not realistic, or study multi-target cases with the assumption that targets have already been detected \cite{ICI_OFDM_TSP_2020}, without addressing the problem of \textit{detection} under strong ICI, which is challenging due to high side-lobe levels. Moreover, all the existing studies investigate the ICI effect in single-input single-output (SISO) OFDM radar configurations; hence, potential benefits that can be provided by a MIMO-OFDM architecture in tackling the ICI problem have been unexplored. Finally, exploitation of ICI has been considered only in scenarios with limited practical relevance (e.g., containing a single target, with restricted transmit symbols \cite{Firat_OFDM_2012,ICI_OFDM_radar_2015,tian2017high,lellouch2014impact}). In a general multi-target setting, ICI conveys crucial information on target velocities and must be exploited to enhance radar performance. In light of the existing literature on OFDM radar sensing under the effect of ICI, three fundamental questions arise:
\begin{itemize}
    \item How can ICI effect be \textit{mitigated} in a generic multi-target scenario with arbitrary transmit data symbols, without hampering the communication capability of the OFDM waveform (i.e., no restrictions on data symbols)?
    \item In what ways can ICI be \textit{exploited} to improve radar performance in a multi-target scenario?
    \item How can we best leverage MIMO architectures to design ICI-aware multi-target detection/estimation schemes?
\end{itemize}

With the goal of answering these questions, this paper tackles the problem of radar sensing with MIMO-OFDM DFRC systems in the presence of non-negligible ICI generated by high-speed targets. 
Motivated by high-mobility mmWave vehicular applications \cite{SPM_JRC_2019}, our goal is to simultaneously \textit{mitigate} and \textit{exploit} ICI associated with multiple targets while retaining the communication functionality of OFDM. Towards that goal, we establish a novel method for ICI-aware sensing in MIMO-OFDM DFRC systems, considering arbitrary transmit symbols in multi-target scenarios. Specifically, we propose an ICI-aware multi-target detection and delay-Doppler-angle estimation algorithm for radar sensing by developing an APES-like spatial filtering approach \cite{Est_MIMO_radar_2008}, coupled with a generalized likelihood ratio test (GLRT) detection procedure. The key idea is to re-formulate radar sensing as a joint carrier frequency offset (CFO)\footnote{Referring to the OFDM communications literature \cite{multiCFO_TWC_2018}, we draw parallelism between ICI-aware sensing and CFO estimation in communications.}  and channel estimation problem, which allows us to decontaminate the ICI effect from the resulting channel estimates, leading to improved performance in target detection and delay-Doppler estimation. Regarding the ICI exploitation aspect, the proposed algorithm enables accurate estimation of any practically relevant unambiguous velocity that well exceeds the standard maximum limit (e.g., in \cite{RadCom_Proc_IEEE_2011,braun2014ofdm}) determined by OFDM symbol duration, thereby effectively \textit{turning ICI from foe to friend}. 
The main contributions can be summarized as follows:
\begin{itemize}
    \item \textbf{Novel Formulation of ICI-aware Sensing as Communication Channel/CFO Estimation Problem:} We establish an insightful duality between ICI-aware sensing in OFDM radar and joint channel/CFO estimation in OFDM communications \cite{multiCFO_TWC_2018,zhang2014blind}. This enables us to formulate the radar delay-Doppler estimation as a joint channel and CFO estimation problem by transforming the delay-Doppler parameterized radar channel into an unstructured form. The key advantage of this novel formulation is that it can produce almost ICI-free channels specific to each target by decoupling ICI compensation from the subsequent delay-Doppler estimation.
    
    
    \item \textbf{ICI-aware Multi-Target Detector/Estimator Design:}  Based on the new problem formulation, we design a three-step ICI-aware multi-target detector/estimator. Observing that spatial covariance matrix does not depend on target delays and Dopplers, we first perform angle estimation using the MUSIC high-resolution direction finding algorithm \cite{MUSIC_1986}. Next, to suppress mutual multi-target interferences \cite{Est_MIMO_radar_2013} in the spatial domain, we propose a spatial filtering approach stemming from \cite{Est_MIMO_radar_2008} that performs joint CFO/channel estimation and beamforming design for each estimated target angle separately. To take into account the presence of multiple targets at a given angle, we devise an orthogonal matching pursuit (OMP) procedure that implements iterative interference cancellation, whereby at each iteration the strongest echo is detected via GLRT and its effect is subtracted from the received signal. As a by-product, the GLRT detector provides the ML estimates (MLEs) of radar channel and CFO associated to the strongest target in the corresponding iteration. In the final step of the algorithm, detection and delay-Doppler estimation can be performed using target-specific channel estimates (decontaminated from ICI) by exploiting the OFDM time-frequency structure.
    
    \item \textbf{ICI Exploitation:} The proposed algorithm exploits the ICI information (obtained as CFO estimates at the output of the second step) in two different ways. First, CFO estimates provide unambiguous target velocities not restricted by OFDM symbol duration, which allows us to resolve velocity ambiguity. Second, we lift the dimension of target resolvability by introducing a fourth dimension (CFO) that allows for distinguishing among targets located at the same delay-Doppler-angle cell. 
    
    
    
\end{itemize}
Additionally, extensive simulations carried out under a wide range of signal-to-noise ratios (SNRs) and target velocities show that the proposed approach provides substantial performance improvements over the conventional FFT based method \cite{RadCom_Proc_IEEE_2011,braun2014ofdm} and achieves performance very close to that obtained by using ICI-free radar observations, in terms of detection and range-velocity estimation accuracy, which demonstrates its superior ICI suppression capability. Furthermore, illustrative examples are presented to reveal the ICI exploitation property of the proposed algorithm\footnote{Notations: Uppercase (lowercase) boldface letters are used to denote matrices (vectors). $(\cdot)^{*}$, $(\cdot)^{T}$ and $(\cdot)^{H}$ represent conjugate, transpose and Hermitian transpose operators, respectively. The $\thnew{n}$ entry of a vector $\xx$ is denoted as $\left[\xx\right]_i$, while the $\thnew{(m,n)}$ element of a matrix $\boldX$ is $\left[ \boldX \right]_{m,n}$. $\projrange{\boldX} = \boldX (\boldX^H \boldX)^{-1} \boldX^H$ represents the orthogonal projection operator onto the subspace spanned by the columns of $\boldX$ and $\odot$ denotes the Hadamard product. $\boldzero$ is an all-zeros vector and $\Imatrix$ denotes an identity matrix of appropriate size. $\diag{\xx}$ represents a diagonal matrix with the elements of $\xx$ on the diagonals and $\vecc{\cdot}$ denotes vectorization operator. $\floorr{\cdot}$ is the floor function.}.



\section{OFDM Radar System Model}
Consider an OFDM DFRC transceiver that communicates with an OFDM receiver while concurrently performing radar sensing using the backscattered signals for target detection \cite{RadCom_Proc_IEEE_2011,DFRC_SPM_2019}. The DFRC transceiver is equipped with an $\Ntx$-element transmit (TX) uniform linear array (ULA) and an $\Nrx$-element receive (RX) ULA.
We assume co-located and perfectly decoupled TX/RX antenna arrays so that the radar receiver does not suffer from self-interference due to full-duplex radar operation \cite{Interference_MIMO_OFDM_Radar_2018,RadCom_Proc_IEEE_2011,OFDM_Radar_Phd_2014,80211_Radar_TVT_2018,Adaptive_RadCom_2019}. In this section, we derive OFDM transmit and receive signal models, and formulate the multi-target detection and parameter estimation problem. 



\subsection{Transmit Signal Model}\label{sec_transmit}
We consider an OFDM communication frame consisting of $M$ OFDM symbols, each of which has a total duration of $\Tsym = \Tcp + T$ and a total bandwidth of $N \deltaf = B$. Here, $\Tcp$ and $T$ denote, respectively, the cyclic prefix (CP) duration and the elementary symbol duration, $\deltaf = 1/T$ is the subcarrier spacing, and $N$ is the number of subcarriers \cite{RadCom_Proc_IEEE_2011}. Then, the complex baseband transmit signal for the $\thnew{m}$ symbol is given by
\begin{equation}\label{eq_ofdm_baseband}
\sm(t) = \frac{1}{\sqrt{N}} \sum_{n = 0}^{N-1}  \xnm \, e^{j 2 \pi n \deltaf t} \rect{\frac{t - m\Tsym}{\Tsym}}
\end{equation} 
where $\xnm$ denotes the complex data symbol on the $\thnew{n}$ subcarrier for the $\thnew{m}$ symbol \cite{General_Multicarrier_Radar_TSP_2016}, and $\rect{t}$ is a rectangular function that takes the value $1$ for $t \in \left[0, 1 \right]$ and $0$ otherwise. Assuming a single-stream beamforming model \cite{80211_Radar_TVT_2018,mmWave_JRC_TAES_2019,MIMO_OFDM_Single_Stream}, the transmitted signal over the block of $M$ symbols for $t \in \left[0, M \Tsym \right]$ can be written as 
\begin{equation}\label{eq_passband_st}
\Re \left\{ \ftx \sum_{m = 0}^{M-1} \sm(t) e^{j 2 \pi \fc t} \right\}
\end{equation}
where $\fc$ and $\ftx \in \complexset{\Ntx}{1}$ denote, respectively, the carrier frequency and the TX beamforming vector.

\subsection{Receive Signal Model}\label{sec_radar_rec}
Suppose there exists a point target in the far-field, characterized by a complex channel gain $\alpha$ (including path loss and radar cross section effects), an azimuth angle $\theta$, a round-trip delay $\tau$ and a normalized Doppler shift $\nu = 2 v/c$ (leading to a time-varying delay $\tau(t) = \tau - \nu t$), where $v$ and $c$ denote the radial velocity and speed of propagation, respectively.
In addition, let $\atx(\theta) \in \complexset{\Ntx}{1}$ and $\arx(\theta) \in \complexset{\Nrx}{1}$ denote, respectively, the steering vectors of the TX and RX ULAs, i.e.,
\begin{align}
        \atx(\theta) &= \transpose{ \left[ 1, e^{j \frac{2 \pi}{\lambda} d \sin(\theta)}, \ldots,  e^{j \frac{2 \pi}{\lambda} d (\Ntx-1) \sin(\theta)} \right] } ~, \\
        \arx(\theta) &= \transpose{ \left[ 1, e^{j \frac{2 \pi}{\lambda} d \sin(\theta)}, \ldots,  e^{j \frac{2 \pi}{\lambda} d (\Nrx-1) \sin(\theta)} \right] } ~,
\end{align}
where $\lambda$ and $d = \lambda/2$ denote the signal wavelength and antenna element spacing, respectively. Given the transmit signal model in \eqref{eq_passband_st}, the backscattered signal impinging onto the $\thnew{i}$ element of the radar RX array can be expressed as
\begin{align}\nonumber
    & y_i(t) = \alpha \left[  \arx(\theta) \right]_i \atx^T(\theta)  \ftx \sum_{m = 0}^{M-1} \sm\big(t - \tau(t)\big) e^{-j 2 \pi \fc \tau} e^{j 2 \pi \fc \nu t }.
\end{align}

We make the following standard assumptions: \textit{(i)} the CP duration is larger than the round-trip delay of the furthermost target\footnote{We focus on small surveillance volumes where the targets are relatively close to the radar, such as vehicular applications.}, i.e., $\Tcp \geq \tau$, \cite{Firat_OFDM_2012,OFDM_Radar_Phd_2014,SPM_JRC_2019}, \textit{(ii)} the Doppler shifts satisfy $ \lvert \nu \rvert \ll 1/N$ \cite{Firat_OFDM_2012,ICI_OFDM_TSP_2020}, and \textit{(iii)} the time-bandwidth product (TBP) $B M \Tsym$ is sufficiently low so that the wideband effect can be ignored, i.e., $\sm(t - \tau(t)) \approx \sm(t - \tau)$ \cite{OFDM_ICI_TVT_2017}. Under this setting, sampling $y_i(t)$ at $t = m\Tsym + \Tcp + \ell T / N$ for $\ell = 0, \ldots, N-1$ (i.e., after CP removal for the $\thnew{m}$ symbol) and neglecting constant terms, the time-domain signal received by the $\thnew{i}$ antenna in the $\thnew{m}$ symbol can be written as \cite{ICI_OFDM_TSP_2020}
\begin{align}\label{eq_rec_bb2}
    y_{i,m}[\ell] &= \alpha \left[  \arx(\theta) \right]_i \atx^T(\theta) \ftx  \, e^{j 2 \pi \fc m \Tsym \nu  } e^{j 2 \pi \fc T \frac{\ell}{N} \nu} \\ \nonumber &~~\times \frac{1}{\sqrt{N}}  \sum_{n = 0}^{N-1}  \xnm \, e^{j 2 \pi n \frac{\ell}{N}} e^{-j 2 \pi n \deltaf \tau} ~.
\end{align}

\subsection{Fast-Time/Slow-Time Representation with ICI}
For the sake of convenience, let us define, respectively, the frequency-domain and temporal steering vectors and the ICI phase rotation matrix as
\begin{align} \label{eq_steer_delay}
	\bb(\tau) & \triangleq  \transpose{ \left[ 1, e^{-j 2 \pi \deltaf \tau}, \ldots,  e^{-j 2 \pi (N-1) \deltaf  \tau} \right] } ~, \\ \label{eq_steer_doppler}
	\cc(\nu) & \triangleq \transpose{ \left[ 1, e^{-j 2 \pi f_c \Tsym \nu }, \ldots,  e^{-j 2 \pi f_c (M-1) \Tsym \nu } \right] } ~, \\ \label{eq_ici_D}
	\DD(\nu) &\triangleq \diag{1, e^{j 2 \pi \fc \frac{T}{N} \nu}, \ldots, e^{j 2 \pi \fc \frac{T(N-1)}{N} \nu} } ~.
\end{align}
Accordingly, the radar observations in \eqref{eq_rec_bb2} can be expressed as
\begin{align} \label{eq_ym}
    \yy_{i,m} &= \alpha \, \left[  \arx(\theta) \right]_i \atx^T(\theta) \ftx  \DD(\nu) \FF_N^{H} \Big(\xx_m \odot \bb(\tau) \left[\cc^{*}(\nu)\right]_m  \Big)  
\end{align}
where $\FF_N \in \complexset{N}{N}$ is the unitary DFT matrix with $\left[ \FF_N \right]_{\ell,n} = \frac{1}{\sqrt{N}} e^{- j 2 \pi n \frac{\ell}{N}} $, $\yy_{i,m} \triangleq \left[ y_{i,m}[0] \, \ldots \, y_{i,m}[N-1] \right]^T$ and $\xx_m \triangleq \left[ x_{0,m} \, \ldots \, x_{N-1,m} \right]^T$.

Aggregating \eqref{eq_ym} over $M$ symbols and considering the presence of multiple targets and noise, the OFDM radar signal received by the $\thnew{i}$ antenna over a frame can be written in a fast-time/slow-time compact matrix form as
\begin{align} \label{eq_ym_all_multi}
    \boldY_i = \sum_{k=0}^{K-1} \alpha^{(i)}_k  \underbrace{\DD(\nu_k)}_{\substack{\rm{ICI} } } \FF_N^{H} \Big(\boldX \odot \bb(\tau_k) \cc^{H}(\nu_k) \Big)  + \boldZ_i
\end{align}
for $i = 0,\ldots, \Nrx-1$, where 
\begin{align}
    \alpha^{(i)}_k &\triangleq \alpha_k \, \left[  \arx(\theta_k) \right]_i \atx^T(\theta_k) \ftx ~, \\
    \boldY_i &\triangleq [ \yy_{i,0} \, \ldots \, \allowbreak \yy_{i,M-1} ] \in \complexset{N}{M} ~, \\ \label{eq_boldX_all}
    \boldX &\triangleq \left[ \xx_0 \, \ldots \, \xx_{M-1} \right] \in \complexset{N}{M} ~,
\end{align}
$\{ \alpha_k, \tau_k, \nu_k, \theta_k \}$ are the parameters of the $\thnew{k}$ target and $\boldZ_i \in \complexset{N}{M}$ is the additive noise matrix with $\vecc{\boldZ_i} \sim \mtCN(\boldzero, \allowbreak \sigma^2 \Imatrix ) $. 

In \eqref{eq_ym_all_multi}, each column contains fast-time samples within a particular symbol and each row contains slow-time samples at a particular range bin. The diagonal phase rotation matrix $\DD(\nu)$ quantifies the ICI effect in fast-time domain, leading to Doppler-dependent phase-shifts across fast-time samples of each OFDM symbol, similar to the CFO effect in OFDM communications \cite{Visa_CFO_TSP_2006,multiCFO_TSP_2019}. With this relation to OFDM communications in mind, the parameter $\nu$ in $\DD(\nu)$ is often referred to as CFO throughout the text to make a clear distinction between the effects induced by $\nu$ via fast-time phase rotations in $\DD(\nu)$ and via slow-time phase rotations in $\cc(\nu)$. To visualize the effect of ICI, Fig.~\ref{fig_ici_comp_range_profile} illustrates the range profile of an OFDM radar, obtained using a standard FFT based method \cite{RadCom_Proc_IEEE_2011,braun2014ofdm}. It is seen that ICI manifests itself in the range profile as increased side-lobe levels, degrading detection performance.


\begin{figure}
	\centering
	\includegraphics[width=1\linewidth]{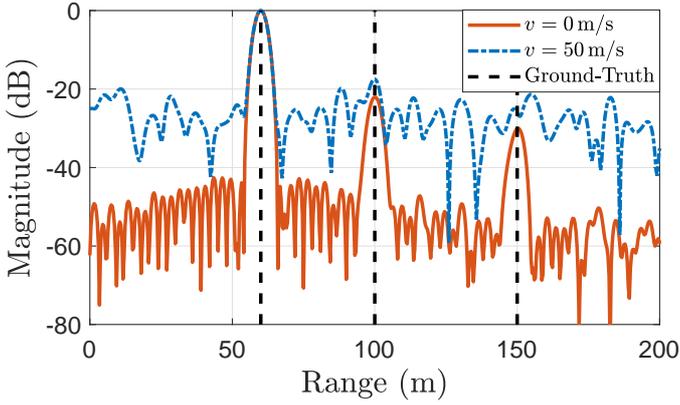}
	\caption{\footnotesize Range profiles of OFDM radar with the parameters given in Table~\ref{tab_parameters} for two different target velocities. The scenario contains $3$ targets having the same velocity $v$, the ranges $(60, 100, 150) \, \rm{m}$, the angles $(25^\circ, 30^\circ, 35^\circ)$ and the SNRs (i.e., $\abs{\alpha_k}^2/\sigma^2$) $(30, 5, 0) \, \rm{dB}$, respectively. In high-mobility scenarios, ICI reduces the dynamic range of OFDM radar due to increased side-lobe levels and causes masking of weak targets.
	}
	\label{fig_ici_comp_range_profile}
	\vspace{-0.1in}
\end{figure}

\subsection{Problem Statement for OFDM Radar Sensing}\label{sec_prob_form}
Given the transmit data symbols $\boldX$, the problem of interest for OFDM radar sensing is to detect the presence of (possibly) multiple targets and estimate their parameters, i.e., channel gains $\{\alpha_k\}_{k=0}^{K-1}$, angles $\{\theta_k\}_{k=0}^{K-1}$, delays $\{\tau_k\}_{k=0}^{K-1}$ and Doppler shifts $\{\nu_k\}_{k=0}^{K-1}$, from the received $\Nrx \times N \times M$ space/fast-time/slow-time data cube $\{\boldY_i\}_{i=0}^{\Nrx-1}$ in \eqref{eq_ym_all_multi}.

\section{ICI-Aware Parameter Estimation via APES Spatial Filtering}\label{sec_ici_aware_apes}
In this section, we propose an ICI-aware delay-Doppler-angle estimation algorithm to tackle the sensing problem formulated in Sec.~\ref{sec_prob_form}. For ease of exposition, we assume the existence of \emph{at most a single target at each azimuth cell}. In Sec.~\ref{sec_glrt_omp}, relying on this approach, we will develop an algorithm that can detect \emph{multiple targets at a given angle} and estimate their parameters. Hence, this section serves as a gentle introduction to the core idea of the paper, which will later be complemented by rigorous detection schemes in Sec.~\ref{sec_glrt_omp}. In the following, we elaborate on the different steps of the proposed algorithm (see Fig.~\ref{fig_stages}).

\begin{figure*}
	\centering
	\includegraphics[width=1\linewidth]{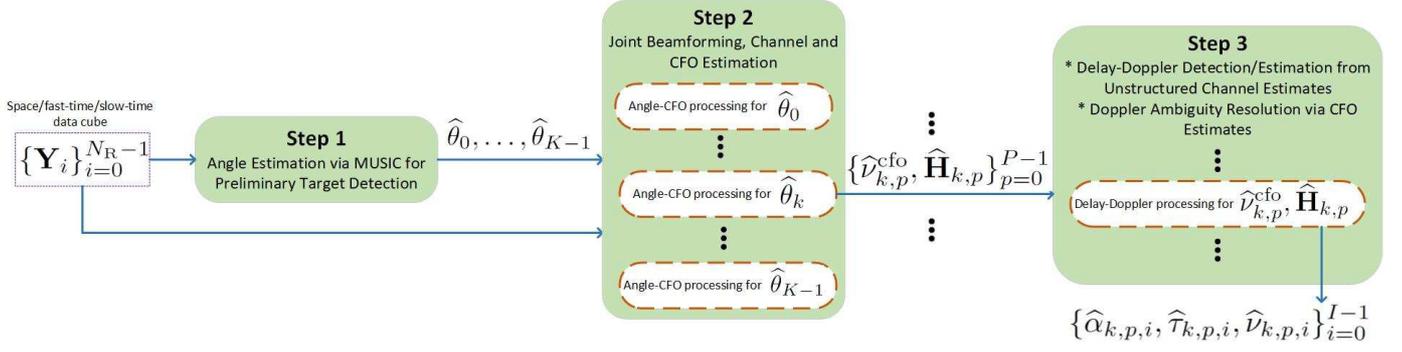}
	\caption{The proposed three-step ICI-aware detector/estimator.}
	\label{fig_stages}
\end{figure*}




\subsection{Step 1: Angle Estimation via MUSIC for Preliminary Target Detection}
In the first step, we wish to identify a set of angles where potential targets may reside so that receive beamformers can be designed accordingly in the subsequent step. For mathematical convenience, we consider the space/fast-time snapshot of the data cube in \eqref{eq_ym_all_multi} corresponding to the $\thnew{m}$ OFDM symbol:
\begin{align}\label{eq_ybbar}
    \Ybbar_m &\triangleq \left[ \yy_{0,m} \, \ldots \, \yy_{\Nrx-1,m}  \right] \in \complexset{N}{\Nrx} \\ \nonumber
    &= \sum_{k=0}^{K-1} \alpha_k \,  \atx^T(\theta_k) \ftx \DD(\nu_k) \FF_N^{H} \diag{\xx_m} \\ \nonumber &~~~~~~~~\times \bb(\tau_k) \left[ \cc^{*}(\nu_k) \right]_m \arx^T(\theta_k)  + \Zbbar_{m}
\end{align}
for $m=0,\ldots,M-1$, where $\Zbbar_{m} \in \complexset{N}{\Nrx}$ is the noise component distributed according to $\vecc{\Zbbar_{m}} \sim \mtCN(\boldzero, \allowbreak \sigma^2 \Imatrix ) $. To distinguish targets with high resolution in the angular domain using small number of antennas (in compliance with mmWave automotive radar requirements \cite{Eldar_SPM_JRC_2020}), we propose to perform angle estimation using the MUSIC algorithm \cite{MUSIC_1986}. To this end, we first construct the spatial covariance matrix (SCM) of the data cube in \eqref{eq_ybbar} as
\vspace{-0.05in}
\begin{equation}\label{eq_SCM_obs}
    \boldR \triangleq \sum_{m=0}^{M-1} \Ybbar_m^H \Ybbar_m ~.
\end{equation}
The following lemma provides an approximation of $\boldR$ under certain conditions.
\begin{lemma}\label{lemma_scm}
Let the covariance matrix of data symbols in \eqref{eq_boldX_all} be given by
\begin{align}\label{eq_cov_x}
    \Eee \{ \vecc{\boldX} \vecc{\boldX}^H \} = \sigmax^2 \Imatrix ~.
\end{align}
Assume that $N$ and/or $M$ is sufficiently large and targets are non-overlapping in either delay or Doppler, i.e.,
\begin{align}\label{eq_delayDoppler_separation}
    \bb^H(\tau_{k_1}) \bb(\tau_{k_2}) \approx 0 ~~~ {\rm{or}} ~~~ \cc^H(\nu_{k_1}) \cc(\nu_{k_2}) \approx 0
\end{align}
for any $k_1 \neq k_2$. Then, the SCM in \eqref{eq_SCM_obs} can be modeled as
\begin{align}
    \boldR &= NM \sigmax^2 \sum_{k=0}^{K-1} \beta_k \, \arx^{*}(\theta_k) \arx^T(\theta_k) + N M \sigma^2 \Imatrix ~,
\end{align}
where $\beta_k \triangleq \abs{\alpha_k}^2 \abs{\atx^T(\theta_k)\ftx}^2$.
\end{lemma}
\begin{proof}
    See Sec.~\ref{sec_app_scm} in the supplementary material.
\end{proof}

Based on Lemma~\ref{lemma_scm}, we observe that the SCM of OFDM radar observations in the presence of ICI is independent of target delays and Dopplers, and follows a standard structure that involves a low-rank (rank-$K$) signal covariance term and a scaled diagonal noise covariance component \cite{MUSIC_1986}. Hence, the standard MUSIC algorithm can be applied.
Assuming $\Nrx > K$, let the eigendecomposition of the SCM be denoted as $\boldR = \boldU_s \boldLambda_s \boldU^H_s + \boldU_n \boldLambda_n \boldU^H_n$, where the diagonal matrix $\boldLambda_s$ contains the $K$ largest eigenvalues, $\boldLambda_n$ contains the remaining $\Nrx-K$ eigenvalues, and $\boldU_s$ and $\boldU_n$ have the corresponding eigenvectors as their columns. Then, the MUSIC spectrum can be computed as 
\vspace{-0.1in}
\begin{align} \label{eq_spatial_spectrum}
    f(\theta) &= \frac{ 1  }{   \arx^T(\theta) \boldU_n \boldU_n^H  \arx^{*}(\theta) } ~.
\end{align}
Let $\Smcal = \{ \thetahat_0, \ldots, \thetahat_{K-1} \}$ be the set of estimated angles in Step~1, which correspond to the peaks of the MUSIC spectrum.

\subsection{Step 2: Angle-Constrained Joint CFO and Unstructured Channel Estimation via APES Beamforming}\label{sec_apes_beamform}
In Step~2, we formulate a joint CFO and channel estimation problem for each $\thetahat \in \Smcal$ determined in Step 1. Invoking the assumption of spatially non-overlapping targets, we treat interferences from other target components as noise and consider a single-target model in \eqref{eq_ybbar} for each $\thetahat \in \Smcal$. To that aim, let
\begin{equation}\label{eq_H_dec}
    \boldH = \left[ \hh_0 \, \ldots \, \hh_{M-1} \right]  \in \complexset{L}{M}
\end{equation}
denote the unstructured, single-target radar channels in the time domain with $L$ taps, collected over $M$ OFDM symbols. Here, $L \leq N \Tcp / T$ due to the CP requirement. Based on this unstructured representation, \eqref{eq_ybbar} can be re-written as
\begin{align} \label{eq_ym_all_single2}
    \Ybbar_m =  \DD(\nu) \Xbbar_m \hh_m  \arx^T(\thetahat)  + \Zbbar_{m} 
\end{align}
where 
\begin{align}
    \Xbbar_m \triangleq \FF_N^{H} \diag{\xx_m} \FF_{N,L} ~,    
\end{align}
$\FF_{N,L} \in \complexset{N}{L}$ denotes the first $L$ columns of $\FF_N$ and $\Zbbar_{m}$ contains noise and interferences from other targets in $\Smcal$.
According to \eqref{eq_ybbar}, the frequency-domain radar channels have the form
\begin{equation}\label{eq_H_def}
    \FF_{N,L} \boldH = \alphabar \, \bb(\tau) \cc^{H}(\nu) 
\end{equation}
with $\alphabar \triangleq \alpha \, \atx^T(\thetahat) \ftx$ representing the complex channel gain including the transmit beamforming effect. 

\begin{remark}[Duality Between CFO/Channel Estimation in OFDM Communications and ICI-aware Sensing in OFDM Radar]\label{remark_down_comm}
Based on the observation that radar targets can be interpreted as uncooperative users from a communications perspective (as they transmit information to the radar receiver via reflections in an unintentional manner \cite{jointRadCom_review_TCOM,chiriyath2017radar}), we point out an interesting duality between the OFDM radar signal model with ICI in \eqref{eq_ym_all_single2} and an OFDM communications model with CFO (e.g., \cite[Eq.~(5)]{multiCFO_TWC_2018} and \cite[Eq.~(4)]{zhang2014blind}). Precisely, $\DD(\nu)$ represents CFO between the OFDM transmitter and receiver for a communications setup, while it quantifies the ICI effect due to high-speed targets for OFDM radar. Similarly, $\Xbbar_m$ represents data/pilot symbols for communications and probing signals for radar\footnote{For radar sensing, every symbol acts as a pilot due to dual-functional operation on a single hardware platform.}. In addition, $\hh_m$ represents the time-domain channel for communications and the structured (delay-Doppler parameterized) channel for radar.
\end{remark}

In light of Remark~\ref{remark_down_comm}, we re-formulate the radar delay-Doppler estimation problem as a communication channel estimation problem, where the objective is to jointly estimate the unstructured time-domain channels $\boldH$ and the CFO $\nu$ from \eqref{eq_ym_all_single2}. To perform channel estimation in \eqref{eq_ym_all_single2}, we propose an APES-like beamformer \cite{Est_MIMO_radar_2008}
\begin{align} \label{eq_apes}
\mathop{\mathrm{min}}\limits_{\ww, \boldH, \nu} &~~ \sum_{m=0}^{M-1}
\normbig{ \Ybbar_m \wwconj -  \DD(\nu) \Xbbar_m \hh_m  }^2  \\ \nonumber
\mathrm{s.t.}&~~ \ww^H \arx(\thetahat) = 1 ~,
\end{align}
where $\ww \in \complexset{\Nrx}{1}$ is the APES spatial beamforming vector for an estimated angle $\thetahat \in \Smcal$. The rationale behind the proposed cost function in \eqref{eq_apes} is to design the beamformer such that the resulting observations $ \big\{ \Ybbar_m \wwconj \big\}_{m=0}^{M-1} $ are as close as possible to the noiseless part of the received signal in \eqref{eq_ym_all_single2}, i.e., $ \big\{ \DD(\nu) \Xbbar_m \hh_m  \arx^T(\thetahat) \wwconj \big\}_{m=0}^{M-1}$. The optimal channel estimate for the $\thnew{m}$ symbol in \eqref{eq_apes} for a given $\ww$ and $\nu$ is given by
\begin{equation}\label{eq_hm_est}
    \hhhat_m = \Big( \Xbbar_m^H \Xbbar_m \Big)^{-1} \Xbbar_m^H \DD^H(\nu) \Ybbar_m \wwconj ~.
\end{equation}
Plugging \eqref{eq_hm_est} back into \eqref{eq_apes} yields
\begin{align} \label{eq_apes_sec3}
\mathop{\mathrm{min}}\limits_{\ww, \nu} &~~ \ww^T \boldQ(\nu) \wwconj ~~~~~~ \mathrm{s.t.}~~ \ww^H \arx(\thetahat) = 1 ~,
\end{align}
where
\vspace{-0.1in}
\begin{align} \label{eq_residual_SCM_CFO}
    \boldQ(\nu) \triangleq \sum_{m=0}^{M-1} \Ybbar_m^H \DD(\nu) \projnull{\Xbbar_m} \DD^H(\nu) \Ybbar_m 
\end{align}
is the null space SCM as a function of CFO, i.e., the SCM of the CFO compensated observations projected onto the null space of the pilot matrices.
For a given CFO $\nu$, the optimal beamformer in \eqref{eq_apes_sec3} can be obtained in closed-form as \cite{Est_MIMO_radar_2008}
\begin{equation}\label{eq_what_sec3}
    \wwhat = \frac{ \boldQ^{*}(\nu)^{-1} \arx(\thetahat) }{ \arx^H(\thetahat) \boldQ^{*}(\nu)^{-1} \arx(\thetahat)  }~.
\end{equation}
Substituting \eqref{eq_what_sec3} into \eqref{eq_apes_sec3}, the CFO can be estimated as
\begin{equation}\label{eq_nuhat}
    \nuhatcfo = \arg \max_{\nu} ~~ \arx^H(\thetahat) \boldQ^{*}(\nu)^{-1} \arx(\thetahat) ~.
\end{equation}
Finally, plugging \eqref{eq_what_sec3} and \eqref{eq_nuhat} into \eqref{eq_hm_est}, the channel estimates can be expressed as
\begin{equation}\label{eq_hm_est2}
    \hhhat_m = \frac{ \Big( \Xbbar_m^H \Xbbar_m \Big)^{-1} \Xbbar_m^H \DD^H(\nuhatcfo) \Ybbar_m  \boldQ(\nuhatcfo)^{-1} \arx^{*}(\thetahat) }{ \arx^T(\thetahat) \boldQ(\nuhatcfo)^{-1} \arx^{*}(\thetahat)   } ~.
\end{equation}
The outputs of Step~2 are the CFO estimate $\nuhatcfo$ in \eqref{eq_nuhat} and the channel estimates $\boldHhat \triangleq \left[ \hhhat_0, \ldots, \hhhat_{M-1} \right]$ in \eqref{eq_hm_est2}.

\subsection{Step 3: Angle-Constrained Delay-Doppler Recovery from Unstructured Channel Estimates}\label{sec_step3_single}
The Step~3 of the proposed algorithm consists of two substeps, as detailed in the following.

\subsubsection{Delay-Doppler Estimation from $\boldHhat$}
Given the unstructured channel estimates in \eqref{eq_hm_est2}, we aim to estimate channel gain, delay and Doppler shift via a least-squares (LS) approach by exploiting the structure in \eqref{eq_H_def} as follows:
\begin{align} \label{eq_apes_step2}
\mathop{\mathrm{min}}\limits_{\alpha, \tau, \nu} &~~ 
\normbig{ \FF_{N,L} \boldHhat -  \alphabar \, \bb(\tau) \cc^H(\nu)  }_F^2 ~.
\end{align}
In \eqref{eq_apes_step2}, delay and Doppler estimates $\tauhat$ and $\nuhat$ can be obtained simply via 2-D FFT, i.e.,
\begin{align}\label{eq_2d_fft_step3}
    (\tauhat, \nuhat) = \arg \max_{\tau, \nu} ~~ \absbig{\bb^H(\tau) \FF_{N,L} \boldHhat \cc(\nu)}^2 ~,
\end{align}
where $\bb(\tau)$ in \eqref{eq_steer_delay} and $\cc(\nu)$ in \eqref{eq_steer_doppler} correspond to DFT matrix columns for a uniform delay-Doppler grid sampled at integer multiples of delay-Doppler resolutions. From \eqref{eq_2d_fft_step3}, channel gain can be estimated as
\begin{align}\label{eq_2d_fft_gain}
    \widehat{\alphabar} = \frac{\bb^H(\tauhat) \FF_{N,L} \boldHhat  \cc(\nuhat)}{ \norm{\bb(\tauhat)}^2 \norm{\cc(\nuhat)}^2 } ~.
\end{align}

\subsubsection{Doppler Ambiguity Resolution via $\nuhatcfo$}
By using the CFO estimate $\nuhatcfo$, we can resolve ambiguity in the Doppler estimate $\nuhat$ in \eqref{eq_2d_fft_step3}. Notice from \eqref{eq_residual_SCM_CFO} and \eqref{eq_nuhat} that $\nuhatcfo$ is estimated based on fast-time phase rotations of $\DD(\nu)$ in \eqref{eq_ici_D}, which implies that the maximum unambiguous CFO that can be estimated via $\DD(\nu)$ is
\vspace{-0.05in}
\begin{align}
    \numaxcfo = \pm \frac{N}{2\fc T}~.
\end{align}
On the other hand, the maximum unambiguous Doppler estimated in \eqref{eq_2d_fft_step3} from slow-time phase rotations of $\cc(\nu)$ in \eqref{eq_steer_doppler} is given by
\vspace{-0.1in}
\begin{align}
    \numax = \pm \frac{1}{2 \fc \Tsym}~,
\end{align}
which is approximately $N$ times smaller than $\numaxcfo$, assuming $\Tcp$ is small compared to $T$. Hence, the ambiguity in $\nuhat$ can be resolved by using $\nuhatcfo$ as
\begin{align}\label{eq_resolve_ambiguity_cfo}
    \nuhat \leftarrow \nuhat + 2 \abs{\numax} \floor{ \frac{\nuhatcfo + \abs{\numax}  }{ 2 \abs{\numax} } } ~.
\end{align}

Referring to the unstructured ML (UML) type methods \cite{swindlehurst1998time,fascista2019millimeter}, we name the proposed algorithm APES-UML, which is summarized in Algorithm~\ref{alg_apes}.

\begin{algorithm}
	\caption{APES-UML for ICI-Aware Sensing with MIMO-OFDM Radar}
	\label{alg_apes}
	\begin{algorithmic}[1]
	    \State \textbf{Input:} Space/fast-time/slow-time data cube $\{\boldY_i\}_{i=0}^{\Nrx-1}$ in \eqref{eq_ym_all_multi}.
	    \State \textbf{Output:} Delay-Doppler-angle-gain estimates of multiple targets $\{ \tauhat_k, \nuhat_k, \thetahat_k, \alphahat_k \}_{k=0}^{K-1}$.
	    \State \textbf{Step~1:} 
	    \begin{enumerate}[label=(\alph*)]
	        \item Estimate target angles by identifying the peaks in the MUSIC spatial spectrum in \eqref{eq_spatial_spectrum}.
	        \end{enumerate}
	    \State \textbf{Step~2:} For each estimated angle $\thetahat$: 
	    \Indent 
	    \begin{enumerate}[label=(\alph*)]
	        \item Estimate the CFO $\nuhatcfo$ via \eqref{eq_nuhat}. 
	        \item Estimate the time-domain channels $\boldHhat$ via \eqref{eq_hm_est2}.
	    \end{enumerate}
	    \EndIndent
        
	    \State \textbf{Step~3:} For each estimated angle $\thetahat$:
	    \Indent
	    \begin{enumerate}[label=(\alph*)]
	        \item Estimate delay-Doppler-gain from the unstructured channel estimates $\boldHhat$ via \eqref{eq_2d_fft_step3} and \eqref{eq_2d_fft_gain}.
	        \item Resolve Doppler ambiguity using $\nuhatcfo$ via \eqref{eq_resolve_ambiguity_cfo}.
	    \end{enumerate}
	    \EndIndent
	    
	\end{algorithmic}
	\normalsize
\end{algorithm}

\section{ICI-Aware Detector/Estimator Design via GLRT and OMP}\label{sec_glrt_omp}
In this section, we extend the APES-UML algorithm proposed in Algorithm~\ref{alg_apes} to the case where multiple targets can be present at an azimuth cell. To accomplish multiple target detection/estimation at a given azimuth angle, we devise an OMP based iterative interference cancellation algorithm using a GLRT detector at each iteration. As we will show, the resulting algorithm will involve replacing in Algorithm \ref{alg_apes}, line 4 with Algorithm~\ref{alg_apes_glrt} and line 5 with Algorithm~\ref{alg_glrt_delayDoppler}. 

\subsection{GLRT for Detection of Multiple Targets at the Same Angle}
Algorithm~\ref{alg_apes} assumes the existence of a single target at each angle estimated in Step~1. To account for the existence of multiple targets at a given angle, the Step~2 of the algorithm can be modified to detect multiple peaks in the CFO spectrum in \eqref{eq_nuhat}. To that end, we design a GLRT detector that, for each $\thetahat \in \Smcal$, operates on the fast-time/slow-time observations $\{ \Ybbar_m^{(p)} \wwconj \}_{m=0}^{M-1}$ obtained by projection of the data cube $\{ \Ybbar_m^{(p)} \}_{m=0}^{M-1}$ onto the fast-time/slow-time domain using a receive beamformer $\ww$ pointing towards $\thetahat$. Here, $\{ \Ybbar_m^{(p)} \}_{m=0}^{M-1}$ denotes the \textit{residue} at the $\thnew{p}$ iteration of the OMP based interference cancellation procedure (whose details will be given in Sec.~\ref{sec_omp_step2}), with the initialization $\Ybbar_m^{(0)} = \Ybbar_m$, where $\Ybbar_m$ is defined in \eqref{eq_ym_all_single2}.
Our goal is to detect the strongest echo at the $\thnew{p}$ iteration and subtract its effect from the current residue $\{ \Ybbar_m^{(p)} \}_{m=0}^{M-1}$. Accordingly, the hypothesis testing problem at the $\thnew{p}$ iteration can be formulated using \eqref{eq_ym_all_single2} as
\begin{align}\label{eq_hypotest}
    \yy = \begin{cases}
	\zz,&~~ {\rm{under~\mathcal{H}_0}}  \\
	\mub(\ww, \boldH, \nu) + \zz,&~~ {\rm{under~\mathcal{H}_1}} 
	\end{cases}
\end{align}
where
\begin{subequations}\label{eq_y_mu}
\begin{align}
\yy &\triangleq \begin{bmatrix} \Ybbar_0^{(p)}  \\ \vdots \\ \Ybbar_{M-1}^{(p)} \end{bmatrix} \wwconj \in \complexset{NM}{1} ~, \\
    \mub(\ww, \boldH, \nu) &\triangleq \begin{bmatrix} \DD(\nu) \Xbbar_0 \hh_0  \\ \vdots \\ \DD(\nu) \Xbbar_{M-1} \hh_{M-1}   \end{bmatrix} \arx^T(\thetahat) \wwconj \in \complexset{NM}{1} ~, \\
    \zz &\triangleq \begin{bmatrix} \Zbbar_{0}  \\ \vdots \\ \Zbbar_{M-1}  \end{bmatrix} \wwconj \in \complexset{NM}{1} ~,
\end{align}
\end{subequations}
with $\mathcal{H}_0$ and $\mathcal{H}_1$ denoting the absence and presence of a target at angle $\thetahat$.

For the composite hypothesis testing problem in \eqref{eq_hypotest} with the unknowns $\ww$, $\boldH$ and $\nu$, the GLRT can be written as
\begin{equation}\label{eq_glrt}
    \llr(\yy) = \frac{ \max_{\ww \in \wtheta, \boldH, \nu} p(\yy \, \lvert \, \mathcal{H}_1 ; \ww, \boldH, \nu ) }{\max_{\ww \in \wtheta} p(\yy \, \lvert \, \mathcal{H}_0 ; \ww ) } \hdet \etatilde
\end{equation}
for some threshold $\etatilde$, where the spatial beamformer steered towards $\thetahat$ is constrained to lie in the set
\begin{align}
    \wtheta = \{ \ww \in \complexset{\Nrx}{1} ~ \lvert ~ \ww^H \arx(\thetahat) = 1 \}~.
\end{align}
Assuming $\zz \sim \mtCN(\boldzero, \sigmabar^2 \Imatrix ) $, the GLRT in \eqref{eq_glrt} takes the form
\begin{equation}\label{eq_glrt2}
    \llr(\yy) = \frac{ \exp\left(- \frac{1}{\sigmabar^2} \min_{\ww \in \wtheta, \boldH, \nu} \norm{  \yy -  \mub(\ww, \boldH, \nu)  }^2 \right)  }{\exp\left(- \frac{1}{\sigmabar^2} \min_{\ww \in \wtheta}   \norm{  \yy }^2  \right) } \hdet \widetilde{\eta} ~.
\end{equation}
By plugging \eqref{eq_y_mu} into \eqref{eq_glrt2} and taking the log, we have
\begin{align}\label{eq_glrt3}
    \llrlog(\yy) &= \frac{1}{\sigmabar^2} \min_{\ww \in \wtheta} \sum_{m=0}^{M-1}
\normbig{ \Ybbar_m^{(p)} \wwconj  }^2 \\ \nonumber & - \frac{1}{\sigmabar^2} \min_{\ww \in \wtheta, \boldH, \nu} \sum_{m=0}^{M-1}
\normbig{ \Ybbar_m^{(p)} \wwconj -  \DD(\nu) \Xbbar_m \hh_m  }^2 \hdet \eta ~,
\end{align}
where $\llrlog(\yy) \triangleq \log \llr(\yy)$ and $\eta \triangleq \log \etatilde$. 

We are now faced with two separate optimization problems to derive the GLRT detector in \eqref{eq_glrt3}. The first problem in \eqref{eq_glrt3} can be re-written as
\vspace{-0.05in}
\begin{align} \label{eq_glrt3_capon}
\mathop{\mathrm{min}}\limits_{\ww \in \wtheta} &~~ \ww^T \boldR^{(p)} \wwconj  ~,
\end{align}
where
\vspace{-0.1in}
\begin{align}
    \boldR^{(p)} \triangleq \sum_{m=0}^{M-1} \big(\Ybbar_m^{(p)}\big)^H \Ybbar_m^{(p)} 
\end{align}
is the SCM of the residue at the $\thnew{p}$ iteration. The problem in \eqref{eq_glrt3_capon} represents a Capon beamforming problem \cite{capon1969high} with the optimal objective value
\begin{align}\label{eq_capon_1st_comp}
     \mathop{\mathrm{min}}\limits_{\ww \in \wtheta} &~~ \ww^T \boldR^{(p)} \wwconj = \frac{1}{\arx^H(\thetahat) \big[ \big(\boldR^{(p)}\big)^{*} \big]^{-1} \arx(\thetahat)} ~.
\end{align}
Regarding the second problem in \eqref{eq_glrt3}, it corresponds to the same APES beamforming problem as investigated in \eqref{eq_apes}. Hence, using the same steps as in \eqref{eq_nuhat}, the optimal CFO for the second optimization in \eqref{eq_glrt3} can be obtained as
\begin{equation}\label{eq_nuhat_grlt3}
    \nuhatcfo_p = \arg \max_{\nu} ~~ \arx^H(\thetahat) \big[ \big(\boldQ^{(p)}(\nu)\big)^{*} \big]^{-1} \arx(\thetahat) ~,
\end{equation}
where 
\vspace{-0.1in}
\begin{align} \label{eq_residual_SCM_CFO_grlt3}
    \boldQ^{(p)}(\nu) \triangleq \sum_{m=0}^{M-1} \big(\Ybbar_m^{(p)}\big)^H \DD(\nu) \projnull{\Xbbar_m} \DD^H(\nu)  \Ybbar_m^{(p)} 
\end{align}
is the null space SCM of the residue at the $\thnew{p}$ iteration. Then, the optimal objective value of the second term in \eqref{eq_glrt3} is given by
\begin{align}\nonumber
    & \min_{\ww \in \wtheta, \boldH, \nu} \sum_{m=0}^{M-1}
\normbig{ \Ybbar_m^{(p)} \wwconj -  \DD(\nu) \Xbbar_m \hh_m  }^2 \\ \label{eq_apes_2nd_comp}
&~~~~~~~~~~~~~~ = \frac{1}{ \arx^H(\thetahat) \big[ \big(\boldQ^{(p)}(\nuhatcfo_p)\big)^{*} \big]^{-1} \arx(\thetahat)  }~.
\end{align}
Finally, inserting \eqref{eq_capon_1st_comp} and \eqref{eq_apes_2nd_comp} into \eqref{eq_glrt3}, the GLRT becomes\footnote{Notice that $\boldR^{(p)} \succeq \boldQ^{(p)}(\nu) $ is satisfied for any $\nu$ since $\boldR^{(p)} = \boldQ^{(p)}(\nu) + \sum_{m=0}^{M-1} \big(\Ybbar_m^{(p)}\big)^H \DD(\nu) \projrange{\Xbbar_m} \DD^H(\nu)  \Ybbar_m^{(p)}$.}
\begin{align}\label{eq_glrt_single}
     \llrlog(\yy) &= \frac{1/\sigmabar^2}{\arx^H(\thetahat) \big[ \big(\boldR^{(p)}\big)^{*} \big]^{-1} \arx(\thetahat)} \\ \nonumber & ~~~~-
      \frac{1/\sigmabar^2}{ \max_{\nu} ~ \arx^H(\thetahat) \big[ \big(\boldQ^{(p)}(\nu)\big)^{*} \big]^{-1} \arx(\thetahat)  } \hdet \eta .
\end{align}

As a summary of the detection part at the $\thnew{p}$ iteration, we perform detection using the GLRT in \eqref{eq_glrt_single} and, if the threshold $\eta$ is crossed, obtain as a by-product the CFO estimate $\nuhatcfo_p$ in \eqref{eq_nuhat_grlt3} associated to the strongest target in the current residue.  
\vspace{-0.1in}

\subsection{OMP for Iterative Interference Cancellation}\label{sec_omp_step2}
Suppose $P$ targets have already been detected at angle $\thetahat$ via GLRT in the previous $P$ iterations, with the corresponding CFO estimates $\{ \nuhatcfo_p \}_{p=0}^{P-1}$. Following an OMP-like procedure \cite{mallat1993matching}, we first update the channel estimates of the $P$ targets detected so far by solving the following optimization problem:
\begin{align} \label{eq_apes_multiple}
\mathop{\mathrm{min}}\limits_{\ww \in \wtheta, \{ \boldH_p \}_{p=0}^{P-1}} &~~ \sum_{m=0}^{M-1}
\norm{ \Ybbar_m \wwconj - \sum_{p=0}^{P-1}  \DD(\nuhatcfo_p) \Xbbar_m \hh_{m,p}  }^2 
\end{align}
where 
\begin{equation}\label{eq_H_dec_p}
    \boldH_p = \left[ \hh_{0,p} \, \ldots \, \hh_{M-1,p} \right]  \in \complexset{L}{M}
\end{equation}
is the channel matrix of the $\thnew{p}$ target. The motivation for the formulation in \eqref{eq_apes_multiple} is to jointly estimate the channels of multiple targets located at angle $\thetahat$ given their CFO parameters by generalizing the APES beamforming problem in \eqref{eq_apes}. Let us define
\begin{align}
    \hhtilde_m^{(P)} &\triangleq \left[ \hh_{m,0}^T \, \ldots \, \hh_{m,P-1}^T \right]^T \in \complexset{LP}{1} \\ \label{eq_phimp}
    \boldPhi_m^{(P)} &\triangleq \left[ \DD(\nuhatcfo_0) \Xbbar_m \, \ldots \, \DD(\nuhatcfo_{P-1}) \Xbbar_m \right] \in \complexset{N}{LP}
\end{align}
for $m = 0, \ldots, M-1$. Note that $\boldPhi_m^{(P)}$ represents the current atom set constructed using the CFO estimates of the detected targets. The problem in \eqref{eq_apes_multiple} can now be written as
\begin{align} \label{eq_apes_multiple2}
\mathop{\mathrm{min}}\limits_{\ww \in \wtheta, \{ \boldH_p \}_{p=0}^{P-1}} &~~ \sum_{m=0}^{M-1}
\norm{ \Ybbar_m \wwconj - \boldPhi_m^{(P)} \hhtilde_m^{(P)}  }^2  ~.
\end{align}
Following similar steps to those in \eqref{eq_apes}--\eqref{eq_hm_est2}, the channel estimates can be obtained from \eqref{eq_apes_multiple2} in closed form as
\begin{align}\label{eq_hhtildehat2}
    \hhtildehat_m^{(P)} = \frac{ \Big[\big(\boldPhi_m^{(P)}\big)^H \boldPhi_m^{(P)} \Big]^{-1} \big(\boldPhi_m^{(P)}\big)^H  \Ybbar_m  \big[ \boldQ^{(P)} \big]^{-1} \arx^{*}(\thetahat) }{ \arx^T(\thetahat) \big[ \boldQ^{(P)} \big]^{-1} \arx^{*}(\thetahat)   }
\end{align}
for $m = 0,  \ldots,  M-1$, where 
\begin{align}
    \boldQ^{(P)} \triangleq \sum_{m=0}^{M-1} \Ybbar_m^H  \projnull{\boldPhi_m^{(P)}}  \Ybbar_m ~.
\end{align}
To ensure linear independence of the columns of $\boldPhi_m^{(P)}$ in \eqref{eq_phimp}, we make the sparsity assumption $P \leq N/L$, i.e., the number of targets at an azimuth cell with distinct CFO values does not exceed $N/L$. This is a typical sparse scene assumption in radar and holds true in general since $N/L = T/\Tcp \gg 1$ for OFDM. Based on the updated channel estimates in \eqref{eq_hhtildehat2}, the residue at the end of the $\thnew{(P-1)}$ iteration is obtained as
\begin{align}\label{eq_residual}
    \Ybbar_m^{(P)} = \Ybbar_m - \boldPhi_m^{(P)} \hhtildehat_m^{(P)} \arx^T(\thetahat) 
\end{align}
for $m = 0,  \ldots,  M-1$. 


As a summary of the OMP based update part, the channel estimates of the $P$ targets detected so far are updated via \eqref{eq_hhtildehat2} and the residue to be used as input for GLRT based detection at the next iteration is computed using \eqref{eq_residual}. The overall algorithm involving GLRT and OMP steps is summarized in Algorithm~\ref{alg_apes_glrt}.

\begin{algorithm}
	\caption{Joint CFO and Radar Channel Estimation with OMP Based Iterative Interference Cancellation}
	\label{alg_apes_glrt}
	\begin{algorithmic}[1]
	    \State \textbf{Input:} Space/fast-time/slow-time data cube $\{\boldY_i\}_{i=0}^{\Nrx-1}$ in \eqref{eq_ym_all_multi}, angle $\thetahat$, maximum number of targets $\Pmax$.
	    \State \textbf{Output:} CFOs and time-domain channel estimates of multiple targets $\{ \nuhatcfo_{p}, \boldHhat_{p} \}_{p=0}^{P-1}$.
	    \State \textbf{Initialization:} Set $P=0$, $\mathcal{A} = \varnothing$ and $\boldPhi_m^{(-1)} = \left[ ~  \right]$.
	    \State \textbf{while} $P < P_{\rm{max}}$
	    \Indent
	    \State Compute $\llrlog(\yy)$ in \eqref{eq_glrt_single}.
	    \State \textbf{if} $\llrlog(\yy) > \eta$
	    \Indent
	    \State Update detected CFOs: $\mathcal{A} \leftarrow \mathcal{A} \cup \{ \nuhatcfo_P \}$.
	    \State Update the atom set: 
	    \vspace{-0.08in}
	    \begin{align}\nonumber
	        \boldPhi_m^{(P)} \leftarrow \left[ \boldPhi_m^{(P-1)} \,  \DD(\nuhatcfo_P) \Xbbar_m \right] ~.
	    \end{align}
	    \State \parbox[t]{210pt}{Update channel estimates of the $P$ targets detected so far based on the updated set of atoms via \eqref{eq_hhtildehat2}.\strut}
	    \State Set $P = P + 1$.
	    \State Update the residual via \eqref{eq_residual}.
	    \EndIndent
	    \State \textbf{else}
	    \Indent 
	    \State \textbf{break}
	    \EndIndent
	    \State \textbf{end if}
	    \EndIndent
	    \State \textbf{end while}
	\end{algorithmic}
	\normalsize
\end{algorithm}

\subsection{GLRT for Detection of Multiple Targets at the Same Angle-CFO Cell}
The previous two subsections focus on Step~2 of Algorithm~\ref{alg_apes} and develop an OMP based interference cancellation procedure to detect multiple targets at a given angle. Similar to Step~2, multiple targets may exist at a given angle-CFO cell, i.e., each channel estimate $\boldHhat$ at the output of Algorithm~\ref{alg_apes_glrt} can be a superposition of echoes of multiple targets at the same CFO $\nuhatcfo$, but with different delays. To handle this case, we propose an extension to Step~3 of Algorithm~\ref{alg_apes} by using a GLRT approach similar to Algorithm~\ref{alg_apes_glrt}.

Suppose that a channel estimate and CFO pair $\{ \nuhatcfo, \boldHhat \}$ is obtained at the output of Algorithm~\ref{alg_apes_glrt}. Based on the structure in \eqref{eq_H_def}, the frequency-domain radar channels in the presence of multiple targets can be modeled as
\begin{align}\label{eq_omp_step3}
    \hhatfs \triangleq \FF_{N,L} \boldHhat = \sum_{i=0}^{I-1} \alpha_{i} \, \bb(\tau_{i}) \cc^{H}(\nu_{i}) + \boldZ ~,
\end{align}
where $\hhatfs \in \complexset{N}{M}$ represents channel estimates in frequency/slow-time domain, $\boldZ \in \complexset{N}{M}$ is the noise term with $\vecc{\boldZ} \sim \mtCN(\boldzero, \allowbreak \sigmatilde^2 \Imatrix ) $ and $I$ is the number of targets located at an angle-CFO cell $(\thetahat, \nuhatcfo)$, with the corresponding delay-Doppler-gain parameters $\{ \tau_{i}, \nu_{i}, \alpha_{i} \}_{i = 0}^{I-1}$. Following a similar approach to Algorithm~\ref{alg_apes_glrt}, we focus on the hypothesis testing problem to test the presence of a single target in \eqref{eq_omp_step3}
\begin{align}\label{eq_hypotest_delayDoppler}
    \hhatfs = \begin{cases}
	\boldZ,&~~ {\rm{under~\mathcal{H}_0}}  \\
	\alpha \, \bb(\tau) \cc^{H}(\nu) + \boldZ,&~~ {\rm{under~\mathcal{H}_1}} 
	\end{cases}~,
\end{align} 
which leads to the GLRT with unknowns $\alpha$, $\tau$ and $\nu$:
\begin{align}\label{eq_glrt_step3}
    \llrlog(\hhatfs) &= \frac{1}{\sigmatilde^2} \normbig{\hhatfs}_F^2 \\ \nonumber &~~~~ - \frac{1}{\sigmatilde^2} \min_{\alpha, \tau, \nu} \normbig{\hhatfs - \alpha \, \bb(\tau) \cc^{H}(\nu) }_F^2   \hdet \eta~.
\end{align}
Note that the second term in \eqref{eq_glrt_step3} has the same form as \eqref{eq_apes_step2}. Hence, using similar steps to those in Sec.~\ref{sec_step3_single}, $\alpha$ can be estimated using \eqref{eq_2d_fft_gain} and the GLRT in \eqref{eq_glrt_step3} becomes
\begin{align}\label{eq_glrt_step3_2}
    \llrlog(\hhatfs) &= \max_{\tau, \nu} \frac{ \absbigs{ \bb^H(\tau) \hhatfs \cc(\nu) }^2}{NM \sigmatilde^2 }  \hdet \eta~.
\end{align}
Contrary to Algorithm~\ref{alg_apes_glrt}, we propose to perform multiple target detection using the metric in \eqref{eq_glrt_step3_2} (which is the output of 2-D FFT, as discussed in Sec.~\ref{sec_step3_single}) by searching for peaks in \eqref{eq_glrt_step3_2} that exceed the threshold $\eta$, without interference cancellation iterations\footnote{The reason is that the resolution of the CFO estimated from $\DD(\nu)$ in \eqref{eq_ici_D} is $1/ (\fc T)$, while the resolution in $\nu$ obtained via $\cc(\nu)$ in \eqref{eq_steer_doppler} is $1 / (\fc M \Tsym)$. Therefore, while targets may interfere with each other in the CFO domain due to poor resolution, the probability of mutual target interference in the delay-Doppler domain is quite low.}. This can be done by a cell-averaging constant false alarm rate (CFAR) detector that operates on the 2-D FFT output in \eqref{eq_glrt_step3_2} \cite[Ch.~6.2.4]{richards2005fundamentals}. Similar to Sec.~\ref{sec_step3_single}, ambiguity in Doppler values of the resulting detections are resolved using $\nuhatcfo$ via \eqref{eq_resolve_ambiguity_cfo}. The overall algorithm for detection of multiple targets residing at the same angle-CFO cell is summarized in Algorithm~\ref{alg_glrt_delayDoppler}.

\begin{algorithm}
	\caption{Delay-Doppler Recovery from Channel Estimates and Doppler Ambiguity Resolution via ICI Exploitation}
	\label{alg_glrt_delayDoppler}
	\begin{algorithmic}[1]
	    \State \textbf{Input:} CFO estimate $\nuhatcfo$, time domain channel estimate $\boldHhat$, probability of false alarm $\pfa$.
	    \State \textbf{Output:} Delay-Doppler-gain estimates of multiple targets $\{ \tauhat_{i}, \nuhat_{i}, \alphahat_{i} \}_{i=0}^{I-1}$.
	    \begin{enumerate}[label=(\alph*)]
	        \item Perform 2-D FFT on $\hhatfs$ in \eqref{eq_omp_step3} to obtain the delay-Doppler spectrum, i.e., GLRT metric in \eqref{eq_glrt_step3_2}.
	        \item Run a cell-averaging CFAR detector with the specified $\pfa$ to detect targets in delay-Doppler domain and estimate their gains via \eqref{eq_2d_fft_gain}.
	        \item For each detected target, use the CFO estimate $\nuhatcfo$ to resolve Doppler ambiguity via \eqref{eq_resolve_ambiguity_cfo}.
	    \end{enumerate}
	\end{algorithmic}
	\normalsize
\end{algorithm}

\section{Numerical Results}\label{sec_sim}
In this section, we evaluate the performance of the proposed ICI-aware sensing algorithm by considering an OFDM system with the parameters specified in Table~\ref{tab_parameters}. With the vehicular JRC scenarios in mind \cite{Eldar_SPM_JRC_2020}, we choose a small number of TX/RX antennas and low bandwidth for low-cost operation. For the signal model in \eqref{eq_ym_all_multi}, the data symbols $\boldX$ are randomly generated from the QPSK alphabet and the transmit beamformer is set to point towards $-30^\circ$, i.e., $\ftx = \atx^{*}(-30^\circ)$. In addition, we define the SNR of a target with reflection coefficient $\alpha_k$ in \eqref{eq_ym_all_multi} as $\snr = \abs{\alpha_k}^2/\sigma^2$. For benchmarking purposes, we compare the following schemes:
\begin{itemize}
    \item \textit{APES-UML:} The proposed ICI-aware sensing algorithm in Algorithm~\ref{alg_apes}.
    \item \textit{2-D FFT:} The standard 2-D FFT method employed in the OFDM radar literature \cite{RadCom_Proc_IEEE_2011,braun2014ofdm}, whose processing chain is as follows. The angles estimated via MUSIC in Step~1 of APES-UML are used to construct receive beamformers and project the data cube $\boldY_i$ in \eqref{eq_ym_all_multi} onto fast-time/slow-time domain, i.e.,
    \begin{align}\label{eq_y_theta}
        \boldY_{\thetahat} = \sum_{i=0}^{\Nrx-1} \boldY_i \big[  \arx^{*}(\thetahat) \big]_i \in \complexset{N}{M} ~,
    \end{align}
    where $\thetahat$ denotes an angle estimated in Step~1. Then, we perform FFT over the columns of $\boldY_{\thetahat}$ in \eqref{eq_y_theta} to obtain frequency/slow-time observations:
    \begin{align}\label{eq_y_theta_freq_time}
        \boldYfs = \FF_N \boldY_{\thetahat}  \in \complexset{N}{M}  ~.
    \end{align}
    Finally, we apply the step~(a) and step~(b) of Algorithm~\ref{alg_glrt_delayDoppler} with $\boldYfs$ in  \eqref{eq_y_theta_freq_time} in place of $\hhatfs$ for target detection.
    
    \item \textit{2-D FFT (ICI-free):} The 2-D FFT method applied on the ICI-free version of the received data in \eqref{eq_ym_all_multi}, i.e.,
    \begin{align} \label{eq_ym_all_multi_ici_free}
    \boldY_i^{{\rm{ICI-free}}} = \sum_{k=0}^{K-1} \alpha^{(i)}_k  \FF_N^{H} \Big(\boldX \odot \bb(\tau_k) \cc^{H}(\nu_k) \Big)  + \boldZ_i ~.
    \end{align}
    This will be used to set an upper bound on the performance of APES-UML.
\end{itemize}
For all the schemes, we employ an identical CFAR detector with the probability of false alarm set as $\pfa = 10^{-4}$.

\begin{table}\footnotesize
\caption{OFDM Simulation Parameters}
\vspace{0.05in}
\centering
    \begin{tabular}{|l|l|}
        \hline
        \textbf{Parameter} & \textbf{Value} \\ \hline
        Carrier Frequency, $\fc$  & $60 \, \rm{GHz}$ \\ \hline
        Total Bandwidth, $B$ & $50 \, \rm{MHz}$ \\ \hline
        Number of Subcarriers, $N$ & $2048$ \\ \hline
        Subcarrier Spacing, $\deltaf$ & $24.41 \, \rm{kHz}$  \\ \hline
        Symbol Duration, $T$ & $40.96 \, \mu \rm{s}$  \\ \hline
        Cyclic Prefix Duration, $\Tcp$ & $10.24 \, \mu \rm{s}$  \\ \hline
        Range Resolution, $\deltar$ & $3 \, \rm{m}$ \\ \hline
        Unambiguous Range, $\Rmax$ & $6144 \, \rm{m}$ \\ \hline
        Maximum Range due to CP, $\Rmax \Tcp/T$ & $1536 \, \rm{m}$ \\ \hline
        Number of Symbols, $M$ & $64$ \\ \hline
        Total Symbol Duration, $\Tsym$ & $51.2 \, \rm{\mu s}$  \\ \hline
        Block Duration, $M \Tsym$  & $3.28 \, \rm{ms}$ \\ \hline 
        Velocity Resolution, $\deltav$ & $0.76 \, \rm{m/s}$ \\ \hline
        Unambiguous Velocity, $\vmax$ & $\pm 24.41 \, \rm{m/s}$  \\ 
        (\textit{Standard})  & \\ \hline
        Unambiguous Velocity, $N \vmax $ & $\pm 62500 \, \rm{m/s}$  \\
         (\textit{ICI Exploitation}) &  \\ \hline
        Number of TX Antennas, $\Ntx$ & $8$ \\ \hline
        Number of RX Antennas, $\Nrx$ & $8$ \\ \hline
    \end{tabular}
    \label{tab_parameters}
    \vspace{-0.1in}
\end{table}

In the following, we first demonstrate the ICI suppression and exploitation capability of the proposed approach via an illustrative example. Then, we assess its detection and estimation performance with respect to benchmark schemes. 

\subsection{Illustrative Example: ICI Suppression and Exploitation Capability of the Proposed Algorithm}
In order to showcase how ICI can be turned from foe to friend using the APES-UML approach in Algorithm~\ref{alg_apes}, we consider a challenging scenario from the perspective of radar detection/estimation, as shown in Fig.~\ref{fig_scenario_config}, where there exist five targets with velocity ambiguities, three of which reside at the same range-velocity-angle cell and two of which are located at the same velocity-angle cell, but with different ranges. In such a scenario, the standard ICI-ignorant OFDM radar algorithms (e.g., \cite{RadCom_Proc_IEEE_2011,braun2014ofdm}) cannot distinguish between Target~1, Target~2 and Target~3 as they fall into the same cell in all three domains. On the contrary, the proposed APES-UML algorithm can resolve these targets via ICI exploitation, as will be shown next through the different steps of Algorithm~\ref{alg_apes}.

\begin{figure*}
	\centering
	\includegraphics[width=1\linewidth]{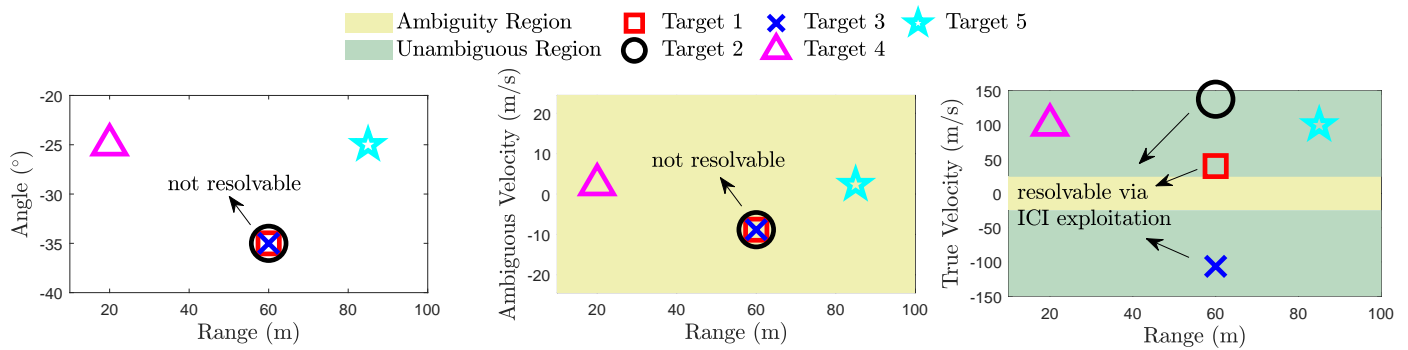}
	\caption{Scenario with multiple targets, illustrated in range-angle, range-ambiguous velocity and range-(true) velocity domains, where target SNRs are given by $\{ 20, 15, 10, 10, -10 \} \, \rm{dB}$, respectively.}
	\label{fig_scenario_config}
\end{figure*}

Fig.~\ref{fig_music_step1} shows the MUSIC spatial spectrum in \eqref{eq_spatial_spectrum} obtained at the output of Step~1, along with the results of ordinary beamforming. As expected, contrary to ordinary beamforming, MUSIC can correctly identify the different target angles with small number of RX antennas, which is crucial for angle-constrained beamforming in Step~2. In Fig.~\ref{fig_cfo_step2}, for each estimated angle $\theta$ in Step~1, we plot the evolution of a normalized version of the GLRT metric in \eqref{eq_glrt_single}, given by
    \begin{align}\label{eq_glrt_single_norm}
     0 \leq 1 - \frac{\arx^H(\theta) (\boldR^{*})^{-1} \arx(\theta)}{\arx^H(\theta) \boldQ^{*}(\nu)^{-1} \arx(\theta)} \leq 1 ~,
\end{align}
with respect to CFO ($\nu$) through successive iterations of the OMP based interference cancellation algorithm in Algorithm~\ref{alg_apes_glrt} (which corresponds to Step~2 of Algorithm~\ref{alg_apes}). For threshold setting in \eqref{eq_glrt_single_norm}, we use a heuristic value of $0.3$ to declare detection. At iteration~$0$, the strongest target at $\theta = -35^{\circ}$, Target~1 with $\snr = 20 \, \rm{dB}$, is detected at the peak of the CFO spectrum in \eqref{eq_glrt_single_norm}. Then, at iteration~1, we can observe the effect of cancelling the interference from Target~1 as a valley in the CFO spectrum centered around the velocity of Target~1. In compliance with the scenario in Fig.~\ref{fig_scenario_config}, the second strongest target at $\theta = -35^{\circ}$, Target~2 with $\snr = 15 \, \rm{dB}$, yields the largest value in the CFO spectrum at iteration~1. It is seen that as the iterations proceed with successive interference cancellation, weaker targets become more pronounced in the CFO spectrum (e.g., Target~3 from iteration~0 to iteration~2), which implies that the proposed OMP based algorithm in Algorithm~\ref{alg_apes_glrt} can successfully eliminate strong target echoes and enable detection of weak targets in the CFO domain. At the final iteration in Fig.~\ref{fig_cfo_step2_1}, the effects of all targets are removed and thus the peak of the CFO spectrum does not exceed the threshold. Similar trends can be observed in Fig.~\ref{fig_cfo_step2_2}, where, for the illuminated angle $\theta = -25^{\circ}$, a single target is detected at $\nu = 100 \, \rm{m/s}$, corresponding to the combined response of Target~4 and Target~5 in Fig.~\ref{fig_scenario_config}.

\begin{figure}
	\centering
	\includegraphics[width=0.9\linewidth]{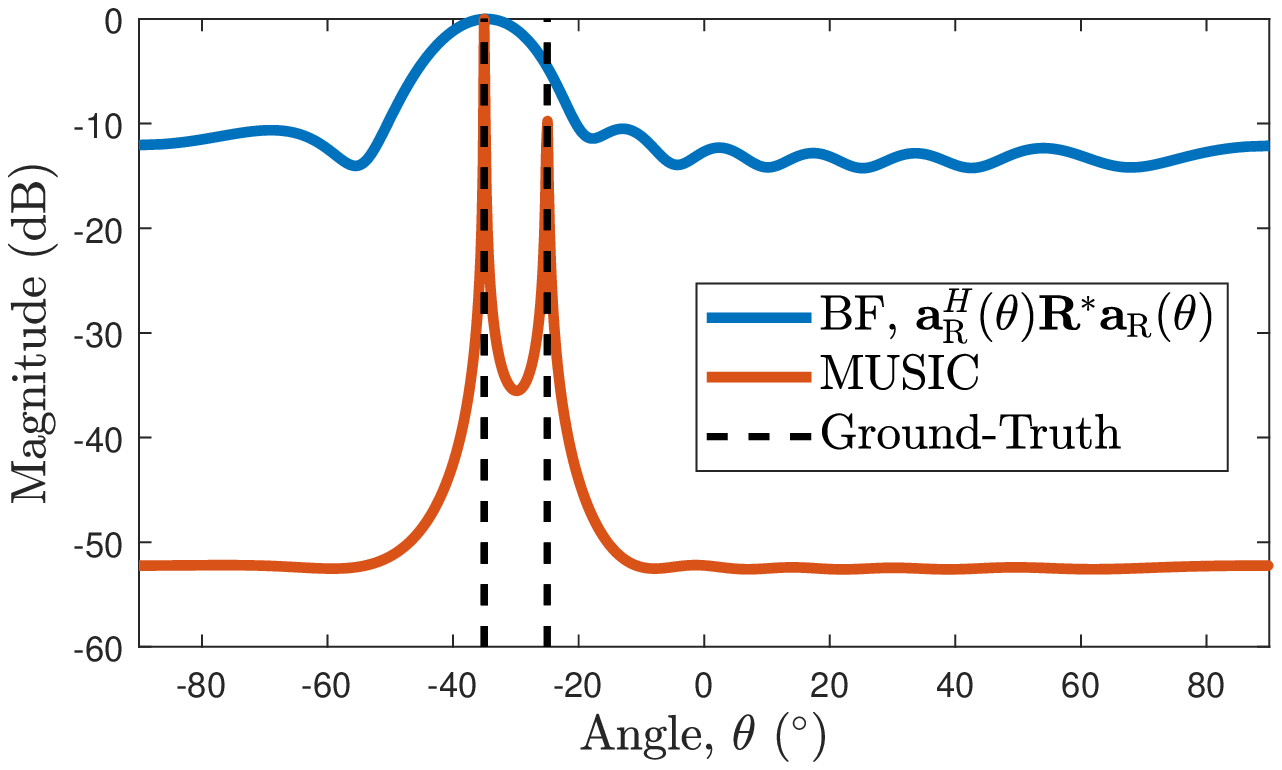}
	\caption{MUSIC spatial spectrum of OFDM radar in Step~1 along with the results of ordinary beamforming (BF) for the scenario in Fig.~\ref{fig_scenario_config}.}
	\label{fig_music_step1}
\end{figure}

\begin{figure}
        \begin{center}
        \subfigure[]{
			 \label{fig_cfo_step2_1}
			 \includegraphics[width=0.47\textwidth]{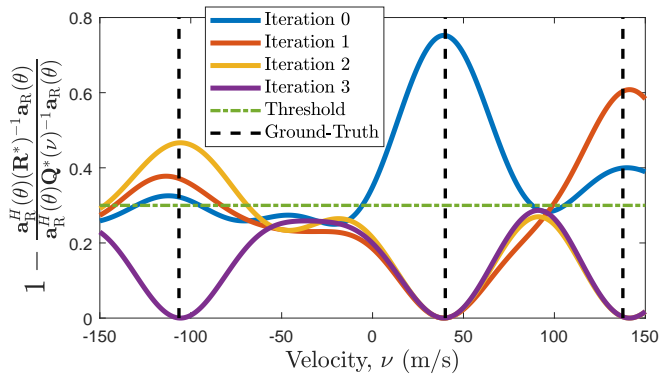}
		}
        \subfigure[]{
			 \label{fig_cfo_step2_2}
			 \includegraphics[width=0.47\textwidth]{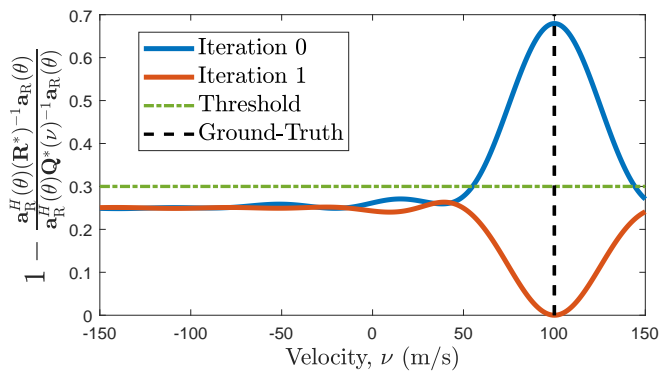}
		}
		
		\end{center}
		\vspace{-0.2in}
        \caption{Evolutions of CFO spectrums (i.e., normalized version of the GLRT metric in \eqref{eq_glrt_single} with respect to CFO $\nu$) obtained for \subref{fig_cfo_step2_1} $\theta = -35^\circ$ and \subref{fig_cfo_step2_2} $\theta = -25^\circ$ through successive iterations of interference cancellation based OMP procedure in Algorithm~\ref{alg_apes_glrt} for the scenario in Fig.~\ref{fig_scenario_config}. As the algorithm proceeds through iterations, the effect of interference cancellation manifests itself as valleys in the CFO spectrum corresponding to the velocity of the strongest target in the corresponding iteration.}  
        \label{fig_cfo_step2}
\end{figure}

The results obtained in Step~2 in Fig.~\ref{fig_cfo_step2} reveal one of the core properties of the proposed ICI-aware sensing algorithm: the \textit{multi-target ICI exploitation capability} with arbitrary transmit symbols. Precisely, the proposed algorithm can resolve Target~1, Target~2 and Target~3 in the CFO domain (indicated as true velocity in the rightmost subfigure in Fig.~\ref{fig_scenario_config}) by exploiting the velocity information conveyed by the ICI effect. As seen from Table~\ref{tab_parameters}, the ICI effect yields an unambiguous velocity that is $N$ times higher than the standard limit (e.g., in \cite{ICI_OFDM_TSP_2020,RadCom_Proc_IEEE_2011,braun2014ofdm}) by virtue of $N$ times faster sampling of fast-time domain compared to slow-time domain in \eqref{eq_ym_all_multi}. Hence, the proposed ICI exploitation approach can distinguish Target~1, Target~2 and Target~3 as separate objects and estimate their true (i.e., unambiguous) velocities. This is only possible through the novel formulation of ICI-aware sensing in Sec.~\ref{sec_ici_aware_apes}, where we decouple the problem of estimating $\nu$ in the fast-time phase rotation matrix $\boldD(\nu)$ from that of estimating $\nu$ in the slow-time steering vector $\cc(\nu)$\footnote{For instance, the algorithm in \cite{ICI_OFDM_TSP_2020} cannot exploit ICI due to coupled estimation of velocity in fast- and slow-time domains.}. On the other hand, the standard 2-D FFT based OFDM radar processing approach \cite{RadCom_Proc_IEEE_2011,braun2014ofdm} only uses slow-time phase progressions for Doppler estimation and thus fails to resolve Target~1, Target~2 and Target~3 since they appear as a unique target in range-angle-ambiguous velocity domains, as shown in Fig.~\ref{fig_scenario_config}.

\begin{figure}
        \begin{center}
        \subfigure[]{
			 \label{fig_range_profile_step3_angle_m35}
			 \includegraphics[width=0.47\textwidth]{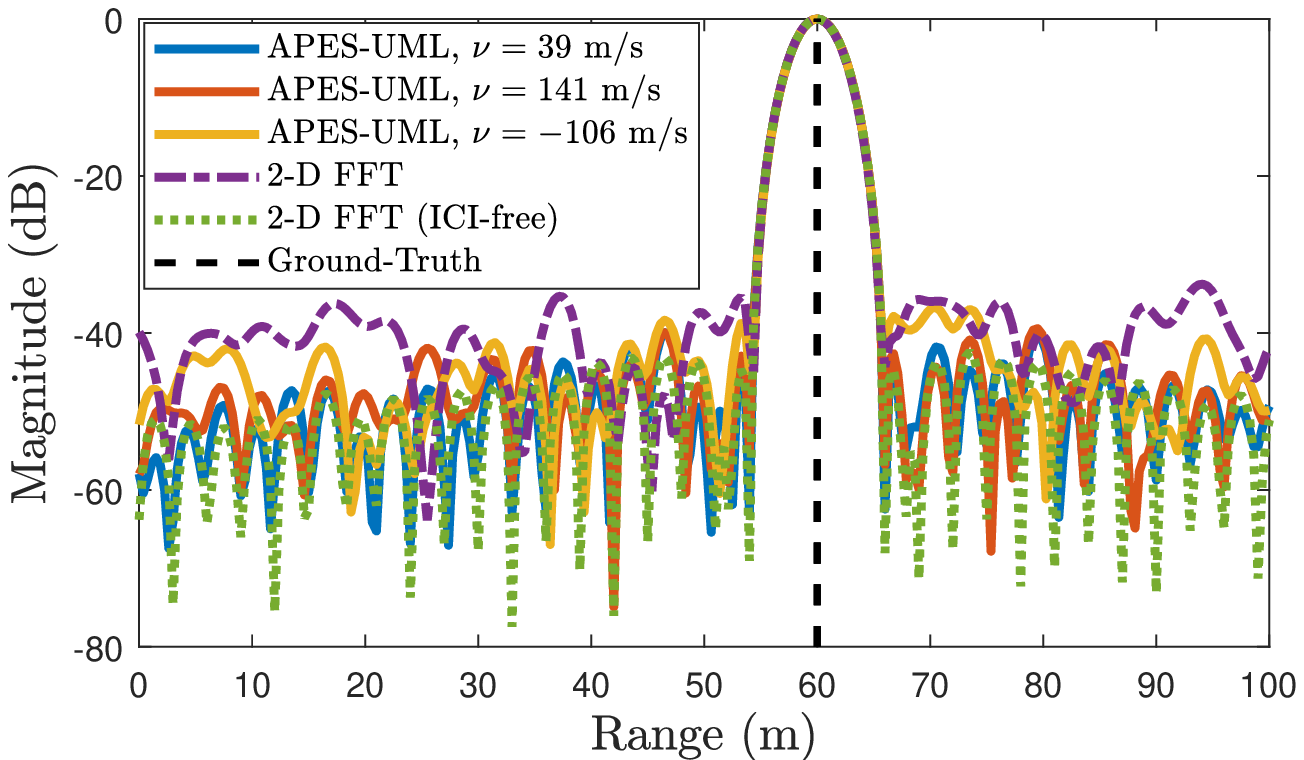}
		}
        \subfigure[]{
			 \label{fig_range_profile_step3_angle_m25}
			 \includegraphics[width=0.47\textwidth]{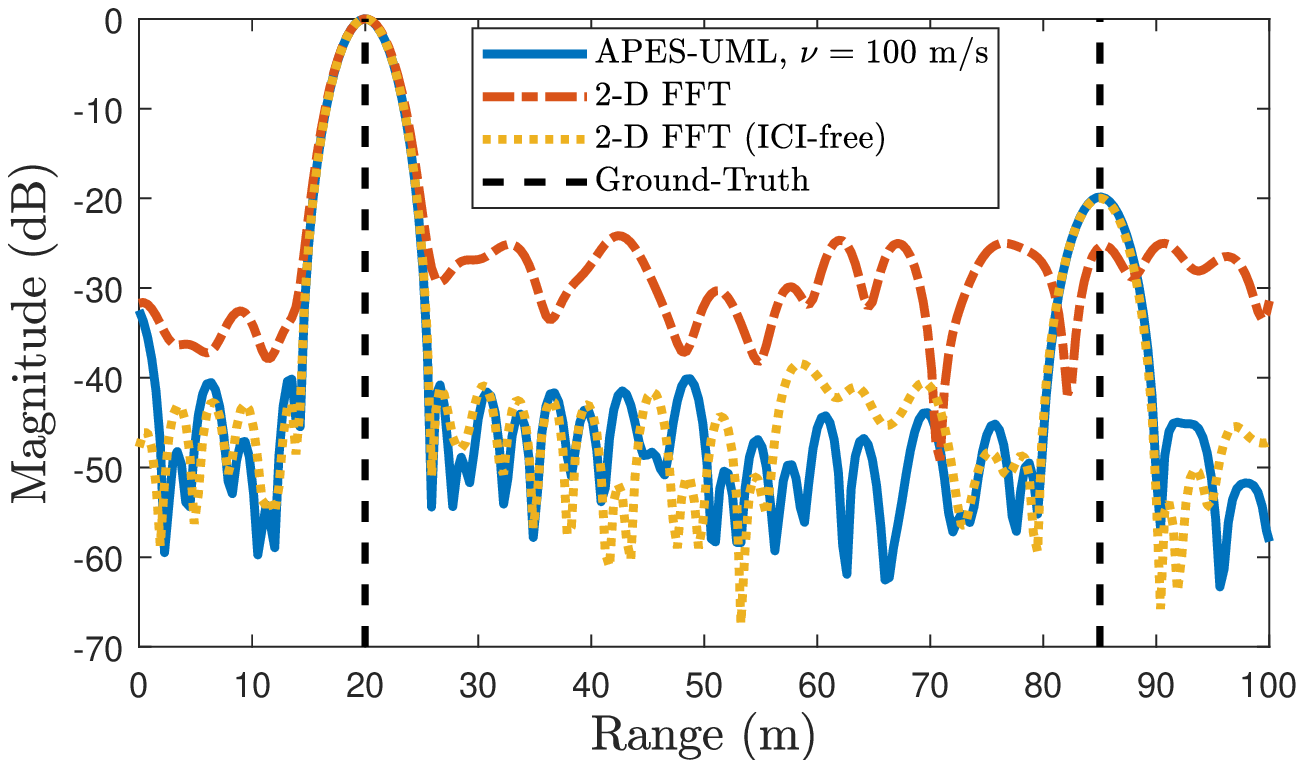}
		}
		
		\end{center}
		\vspace{-0.2in}
        \caption{Range profiles obtained for \subref{fig_range_profile_step3_angle_m35} $\theta = -35^\circ$ and \subref{fig_range_profile_step3_angle_m25} $\theta = -25^\circ$ in Algorithm~\ref{alg_glrt_delayDoppler}, corresponding to CFOs estimated in Fig.~\ref{fig_cfo_step2}, along with the profiles obtained by the FFT benchmarks. When the standard FFT method is employed, the ICI effect leads to masking of Target~5, while APES-UML can successfully eliminate ICI to make the target peak visible as in the ICI-free case.} 
        \label{fig_range_profile_step3}
        \vspace{-0.2in}
\end{figure}

We now investigate the output of Step~3 of Algorithm~\ref{alg_apes} (implemented using Algorithm~\ref{alg_glrt_delayDoppler}), which uses the channel estimates from Step~2 to detect targets in the delay-Doppler domain. Fig.~\ref{fig_range_profile_step3} illustrates the range profiles obtained for different angles and CFOs, estimated in Step~1 and Step~2, respectively, along with the range profiles of the FFT based benchmarks. For $\theta = -35^\circ$, it is observed that as Target~1, Target~2 and Target~3 lie in the same range bin, their respective range profiles have almost an identical shape. Since these targets are already resolved by the proposed APES-UML algorithm in the CFO domain in the previous step, being co-located in the range domain does not have any effect on their detection performance. For $\theta = -25^\circ$, two targets, Target~4 and Target~5, appear at different locations in the range profile of APES-UML corresponding to the CFO, shown in Fig.~\ref{fig_cfo_step2_2}, and the accompanying channel estimate, which is consistent with the scenario in Fig.~\ref{fig_scenario_config}. It is observed from both Fig.~\ref{fig_range_profile_step3_angle_m35} and Fig.~\ref{fig_range_profile_step3_angle_m25} that the ICI effect leads to increased side-lobe levels for the standard 2-D FFT approach. Meanwhile, the proposed APES-UML algorithm can achieve approximately the same side-lobe levels as the ICI-free case, which proves its \textit{multi-target ICI suppression capability} (i.e., it can effectively mitigate the ICI effects caused by multiple targets having different velocities via accurate estimation of their CFOs in Step~2). Moreover, we see that ICI-induced high side-lobe levels leads to masking of Target~5 in the range profile of the FFT method, while APES-UML and FFT in the ICI-free case produce a peak at the location of Target~5. 

\vspace{-0.1in}
With regard to the implications of Fig.~\ref{fig_cfo_step2} and Fig.~\ref{fig_range_profile_step3}, it is worth emphasizing another important property of the proposed method: \textit{spatial filtering}, which is enabled by the MIMO architecture and the APES framework developed in Sec.~\ref{sec_apes_beamform}. By comparing Fig.~\ref{fig_cfo_step2_1} and Fig.~\ref{fig_cfo_step2_2}, we notice that targets located at different angles do not spill much energy into each other's CFO spectrum. Similarly, inspecting Fig.~\ref{fig_range_profile_step3_angle_m35} and Fig.~\ref{fig_range_profile_step3_angle_m25}, no leakage can be observed between the range profiles corresponding to different angles. We accomplish this by designing the APES-like cost function in \eqref{eq_apes}, or equivalently, in \eqref{eq_apes_multiple}, to perform joint optimization of beamformer, CFO and radar channel, which helps suppress energy leakage outside the desired angle. Therefore, the proposed APES-UML approach can separate out individual target reflections in the angular domain from the mixed signal in \eqref{eq_ym_all_multi} by leveraging the multiple-antenna structure.

\subsection{Detection and Estimation Performance}
In this part, we study the detection and estimation performance of the considered OFDM sensing algorithms using $100$ independent Monte Carlo noise realizations. We consider a scenario with two targets as described in Table~\ref{tab_scenario_mc}, where Target~2 is chosen as the reference target to evaluate performance metrics\footnote{For the sake of fairness towards the FFT based benchmarks, detection decisions are based on ambiguous range-velocity values for the FFT based schemes and on true (unambiguous) range-velocity values for the APES-UML algorithm (i.e., contrary to the FFT based methods, APES-UML needs to resolve ambiguities to be able to declare detection).}. The aim is to investigate the masking effect of ICI (due to increased side-lobe levels) under a wide variety of operating conditions, including various SNRs of Target~2 ($\snr \, \rm{dB}$) and target velocities of both targets ($\nu \, \rm{m/s}$) in the presence of a strong target, Target~1. 

\begin{table}\footnotesize
\caption{Scenario with Varying SNR and Target Velocities}
\vspace{0.05in}
\centering
    \begin{tabular}{|l|l|l|l|l|}
        \hline
        & Range & Velocity & Angle & SNR \\ \hline
        Target~1 & $40 \, \rm{m}$ & $\nu \, \rm{m/s}$ & $-35^{\circ}$ & $25 \, \rm{dB}$ \\ \hline
        Target~2 (Reference) & $80 \, \rm{m}$ & $\nu \, \rm{m/s}$ & $-25^{\circ}$ & $\snr \, \rm{dB}$ \\ \hline
    \end{tabular}
    \label{tab_scenario_mc}
    \vspace{-0.1in}
\end{table}

Fig.~\ref{fig_pd_snr} shows the probability of detection of the reference target as a function of SNR for three different target velocities. In agreement with the side-lobe performances in Fig.~\ref{fig_range_profile_step3}, APES-UML significantly outperforms the standard FFT scheme and performs very close to the FFT benchmark that uses ICI-free observations, which proves that the proposed approach can effectively suppress the ICI effect associated with multiple targets. In addition, the detection performance of APES-UML is resilient to target velocity; it can attain the upper bound achievable through ICI-free observations at all target velocities. On the other hand, the performance of 2-D FFT deteriorates as the velocity increases since the ICI effect becomes more severe at higher velocities. At $\nu = 120 \, \rm{m/s}$, an OFDM radar employing standard 2-D FFT processing \cite{RadCom_Proc_IEEE_2011,braun2014ofdm} becomes completely blind within the SNR range of interest, which clearly indicates the significance of ICI-aware sensing in high-mobility scenarios. 

\begin{figure}
        \begin{center}
        \subfigure[]{
			 \label{fig_pd_snr_20mps}
			 \includegraphics[width=0.47\textwidth]{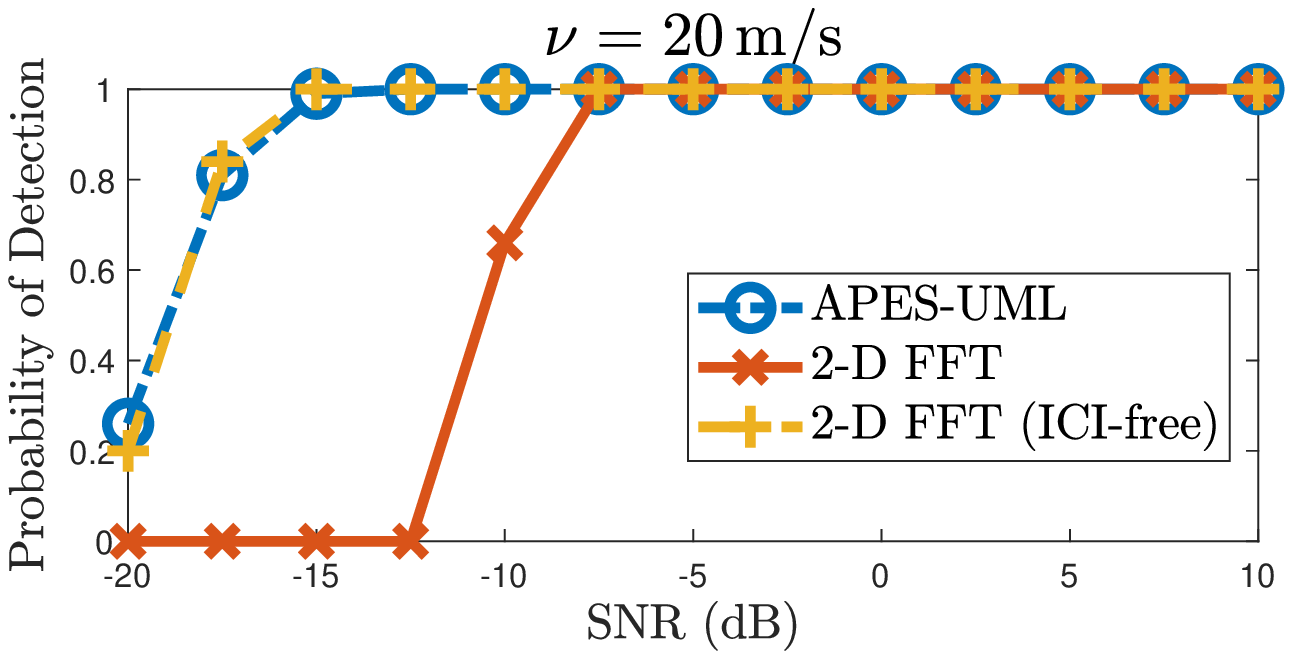}
		}
        \subfigure[]{
			 \label{fig_pd_snr_70mps}
			 \includegraphics[width=0.47\textwidth]{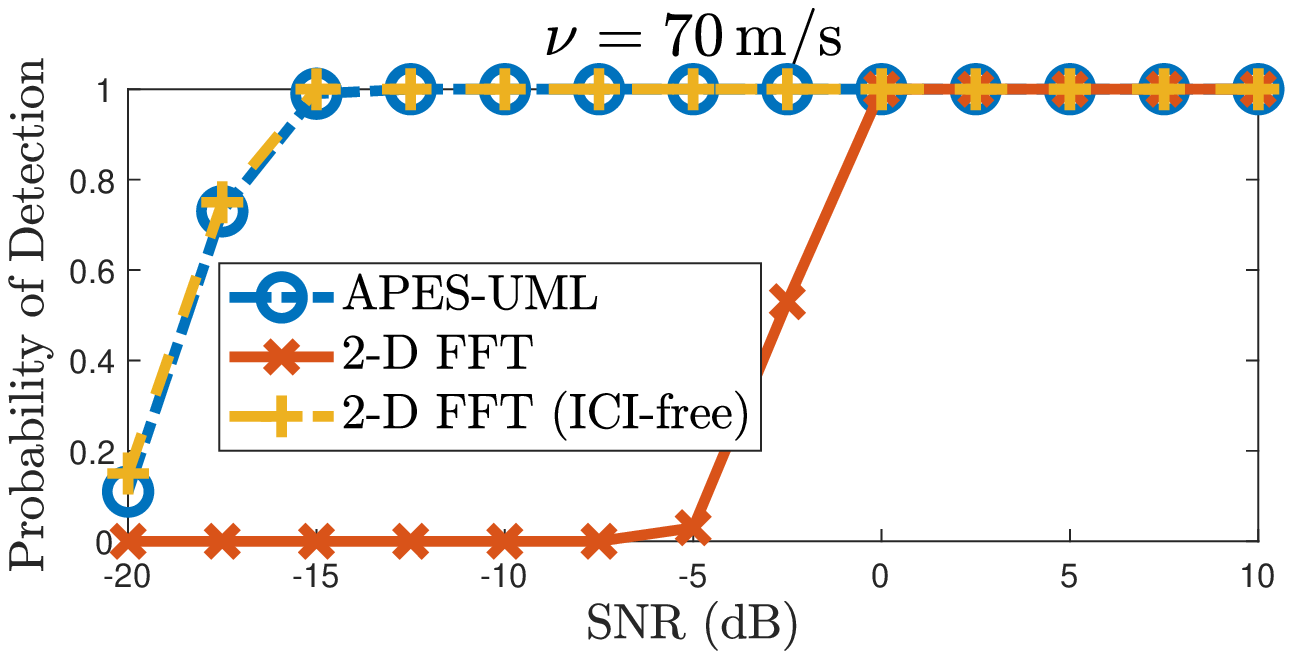}
		}
		
		\subfigure[]{
			 \label{fig_pd_snr_120mps}
			 \includegraphics[width=0.47\textwidth]{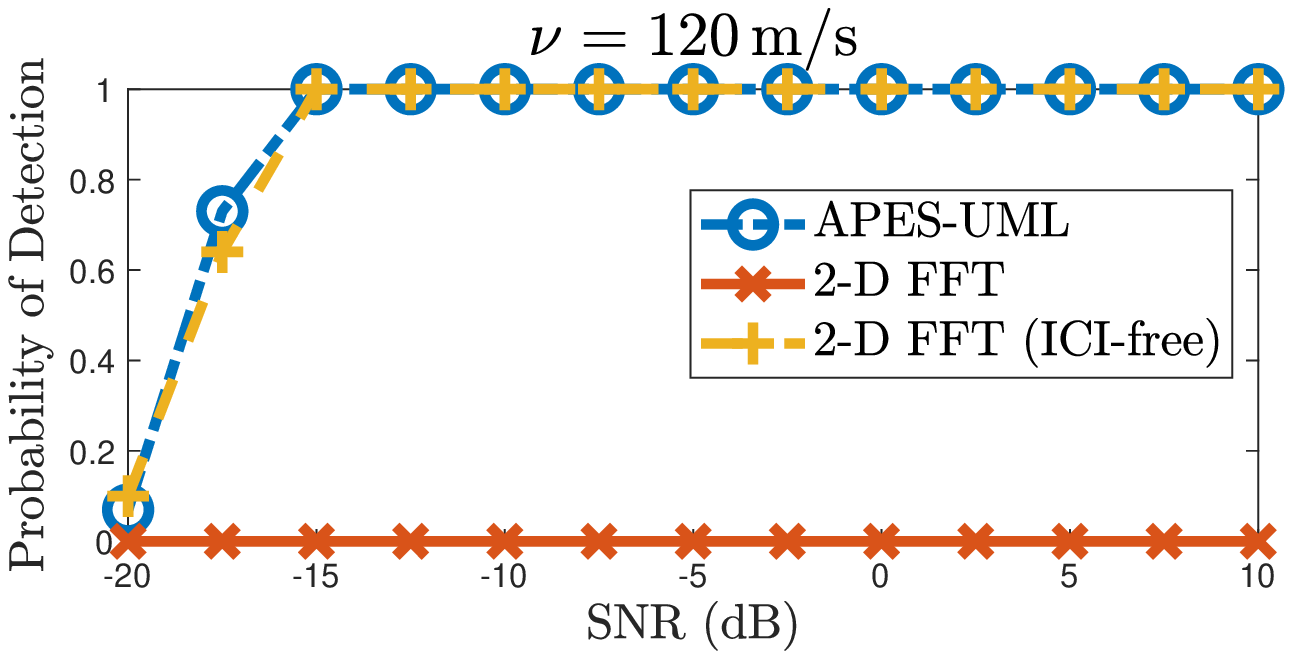}
		}
		
		\end{center}
		\vspace{-0.2in}
        \caption{Probability of detection of the reference target with respect to SNR for \subref{fig_pd_snr_20mps} $\nu = 20 \, \rm{m/s}$, \subref{fig_pd_snr_20mps} $\nu = 70 \, \rm{m/s}$, and \subref{fig_pd_snr_120mps} $\nu = 120 \, \rm{m/s}$ for the scenario in Table~\ref{tab_scenario_mc}.} 
        \label{fig_pd_snr}
        \vspace{-0.1in}
\end{figure}

In Fig.~\ref{fig_pfa_snr}, we examine false discovery rates (FDRs), defined as \cite{fdr_2011} ${\rm{FDR}} = {V}/{(V+S)}$, 
where $V$ is the total number of false alarms (i.e., detections that cannot be associated to the existing targets) and $S$ is the total number of reference target detections over all Monte Carlo runs. Similar to the probability of detection curves, APES-UML exhibits a false alarm performance that is very close to FFT with ICI-free observations, which proves the multi-target ICI compensation capability of the proposed approach. From Fig.~\ref{fig_pd_snr} and Fig.~\ref{fig_pfa_snr}, we conclude that, except for the case of high SNR and low-mobility, the standard FFT method fails.

\begin{figure}
        \begin{center}
        \subfigure[]{
			 \label{fig_pfa_snr_20mps}
			 \includegraphics[width=0.47\textwidth]{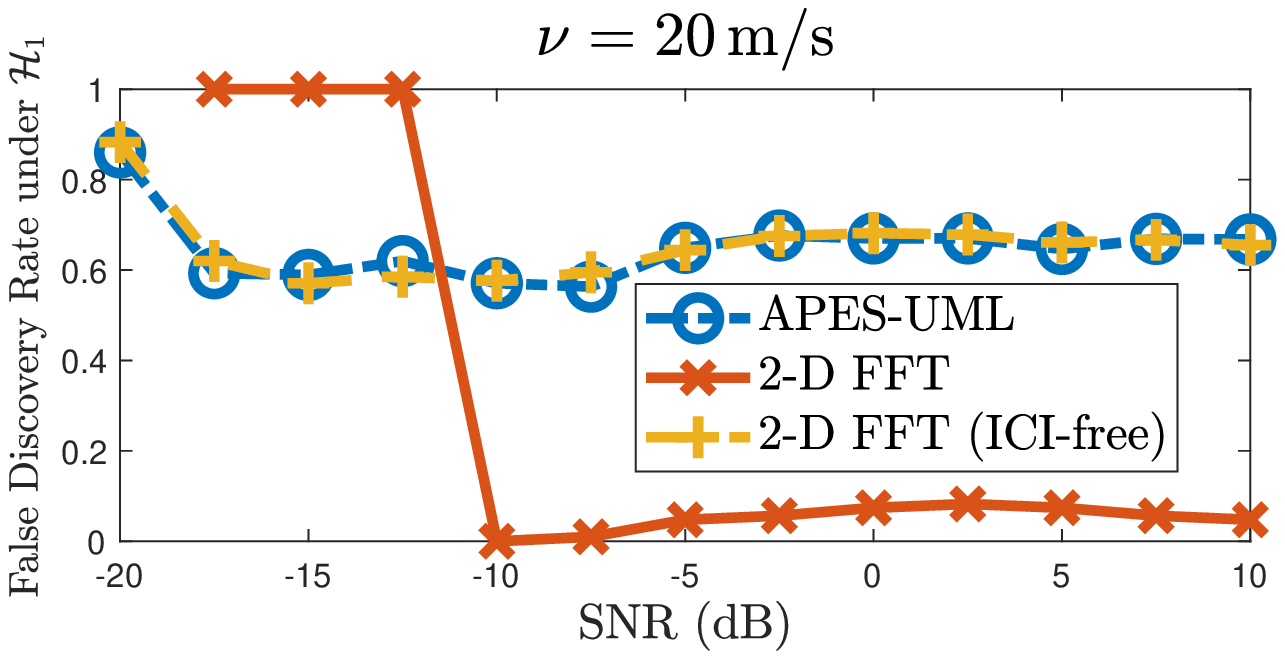}
		}
        \subfigure[]{
			 \label{fig_pfa_snr_70mps}
			 \includegraphics[width=0.47\textwidth]{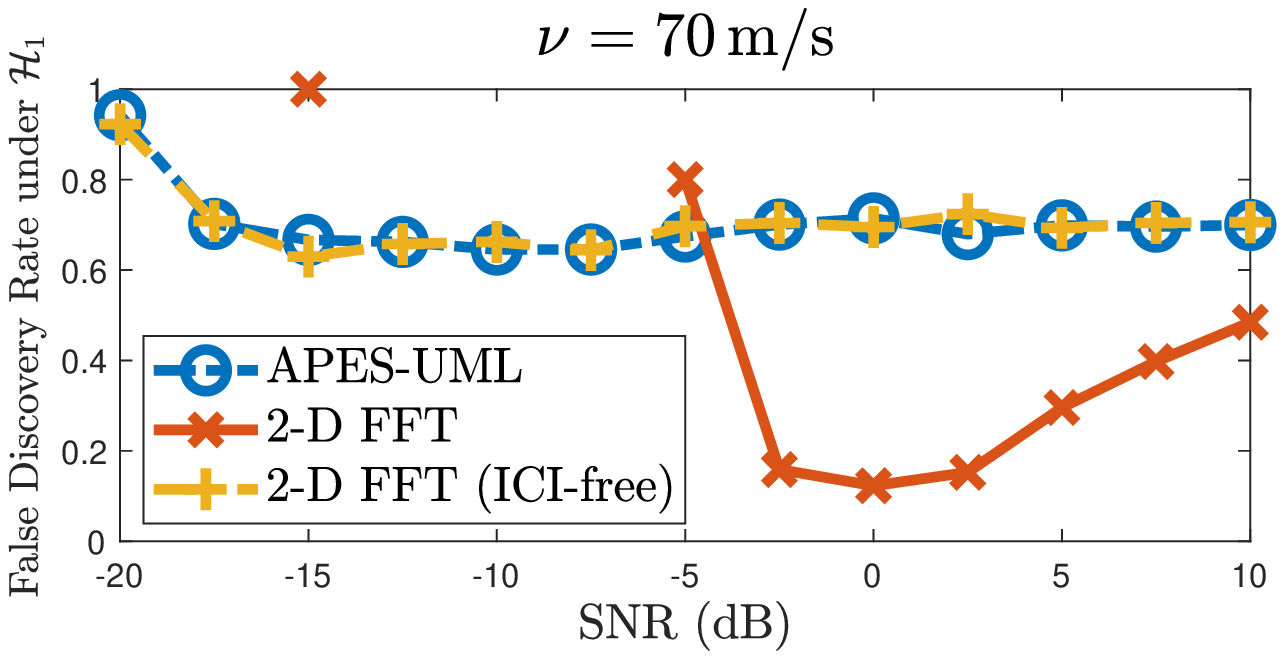}
		}
		
		\subfigure[]{
			 \label{fig_pfa_snr_120mps}
			 \includegraphics[width=0.47\textwidth]{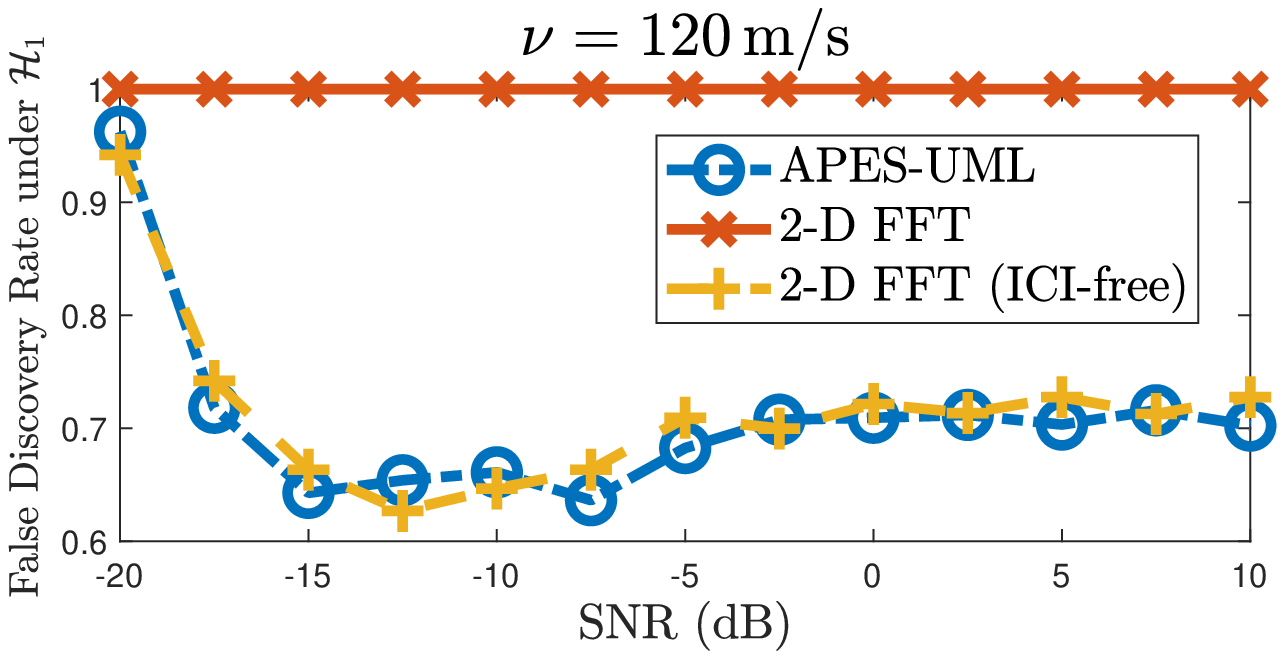}
		}
		
		\end{center}
		\vspace{-0.2in}
        \caption{False discovery rate under $\mathcal{H}_1$ with respect to SNR for \subref{fig_pfa_snr_20mps} $\nu = 20 \, \rm{m/s}$, \subref{fig_pfa_snr_70mps} $\nu = 70 \, \rm{m/s}$, and \subref{fig_pfa_snr_120mps} $\nu = 120 \, \rm{m/s}$ for the scenario in Table~\ref{tab_scenario_mc}. No values are shown if there is neither a detection nor a false alarm.}  
        \label{fig_pfa_snr}
        \vspace{-0.1in}
\end{figure}

We now turn our attention to estimation performances of the considered schemes. Fig.~\ref{fig_range_rmse_snr} shows the root mean-squared errors (RMSEs) of range estimation of the reference target as a function of SNR for various target velocities. The RMSE is calculated as \cite{grossi2020adaptive}
\begin{align}\label{eq_rmse}
    {\rm{RMSE}} = \left( \Eee \Big\{ (\widehat{R} - R)^2 ~ \lvert ~ {\rm{target~detected}} \Big\} \right)^{1/2} ~,
\end{align}
where $\widehat{R}$ and $R$ denote, respectively, the estimated and true range values. As expected, the 2-D FFT benchmark cannot estimate target parameters below a certain SNR threshold (depending on velocity) due to lack of detections, in compliance with Fig.~\ref{fig_pd_snr}. In addition, the proposed APES-UML algorithm achieves almost the same range RMSE performance as the ICI-free benchmark, which again evidences its superior ICI elimination capability. Moreover, APES-UML exhibits consistent range estimation performance at all target velocities; no noticeable changes can be observed in the range RMSE of APES-UML with increasing velocity, whereas the ICI effect significantly degrades the performance of FFT method, especially at high velocities. 

In Fig.~\ref{fig_vel_rmse_snr}, we plot the velocity RMSEs of the reference target with respect to SNR, calculated similarly to \eqref{eq_rmse}. For $\nu = 20 \, \rm{m/s}$, similar trends to the case of range RMSE can be observed. However, for $\nu = 70 \, \rm{m/s}$ and $\nu = 120 \, \rm{m/s}$, both of the FFT-based methods fail to correctly estimate velocity since the unambiguous velocity is $\vmax = \pm 24.41 \, \rm{m/s}$, as seen from Table~\ref{tab_parameters}. Through simultaneous mitigation (via joint CFO/channel estimation in Algorithm~\ref{alg_apes_glrt}) and exploitation (by resolving velocity ambiguity in Algorithm~\ref{alg_glrt_delayDoppler}) of ICI, APES-UML can estimate the true velocity of the target with high accuracy. Hence, the proposed approach can even outperform the ICI-free benchmark in such scenarios by turning ICI from foe to friend.

\begin{figure}
        \begin{center}
        \subfigure[]{
			 \label{fig_range_rmse_snr_20mps}
			 \includegraphics[width=0.47\textwidth]{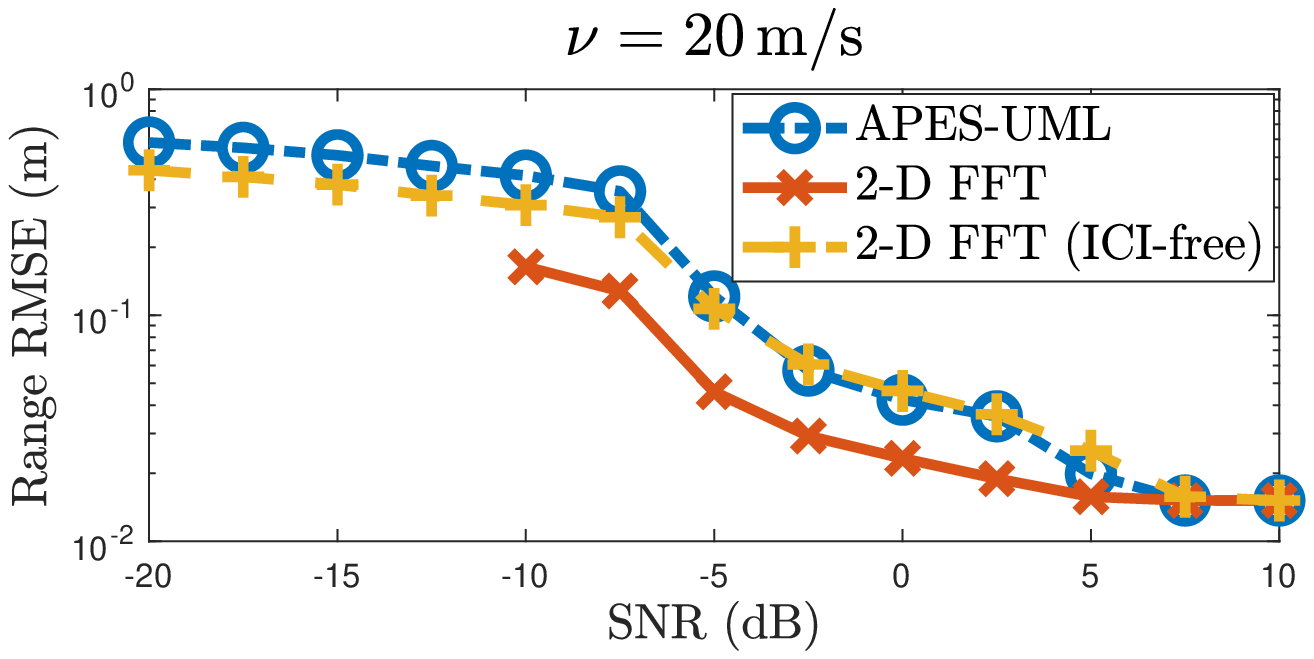}
		}
        \subfigure[]{
			 \label{fig_range_rmse_snr_70mps}
			 \includegraphics[width=0.47\textwidth]{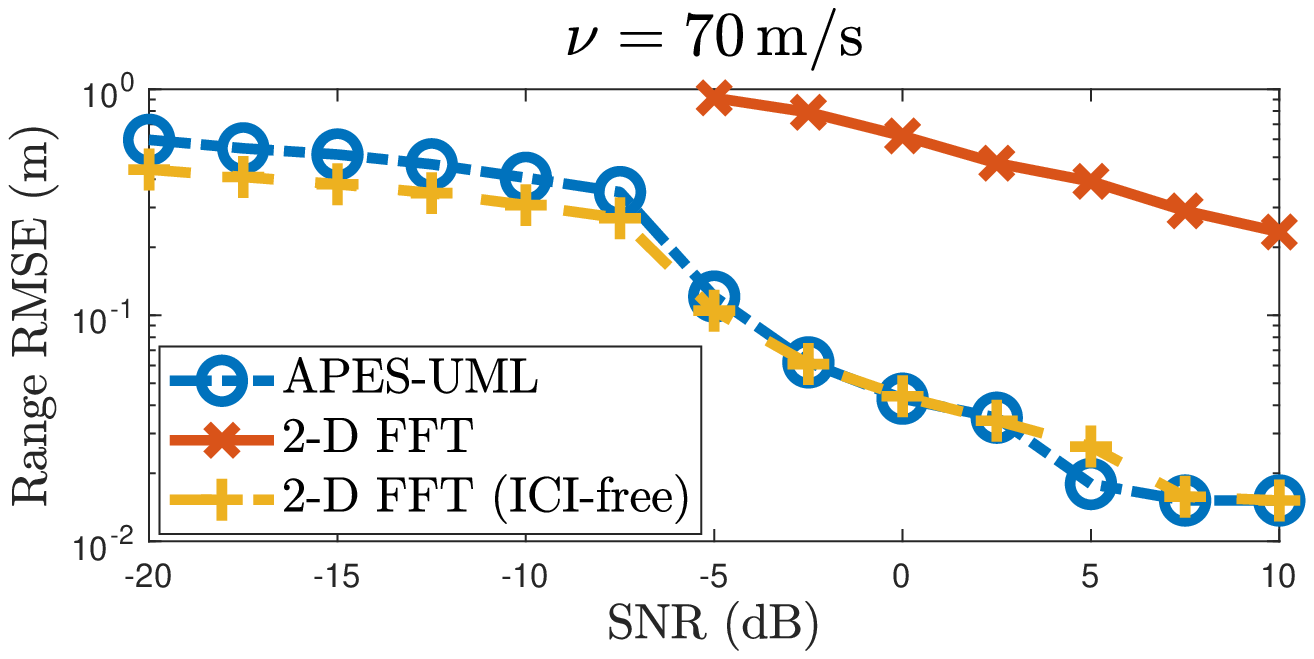}
		}
		
		\subfigure[]{
			 \label{fig_range_rmse_snr_120mps}
			 \includegraphics[width=0.47\textwidth]{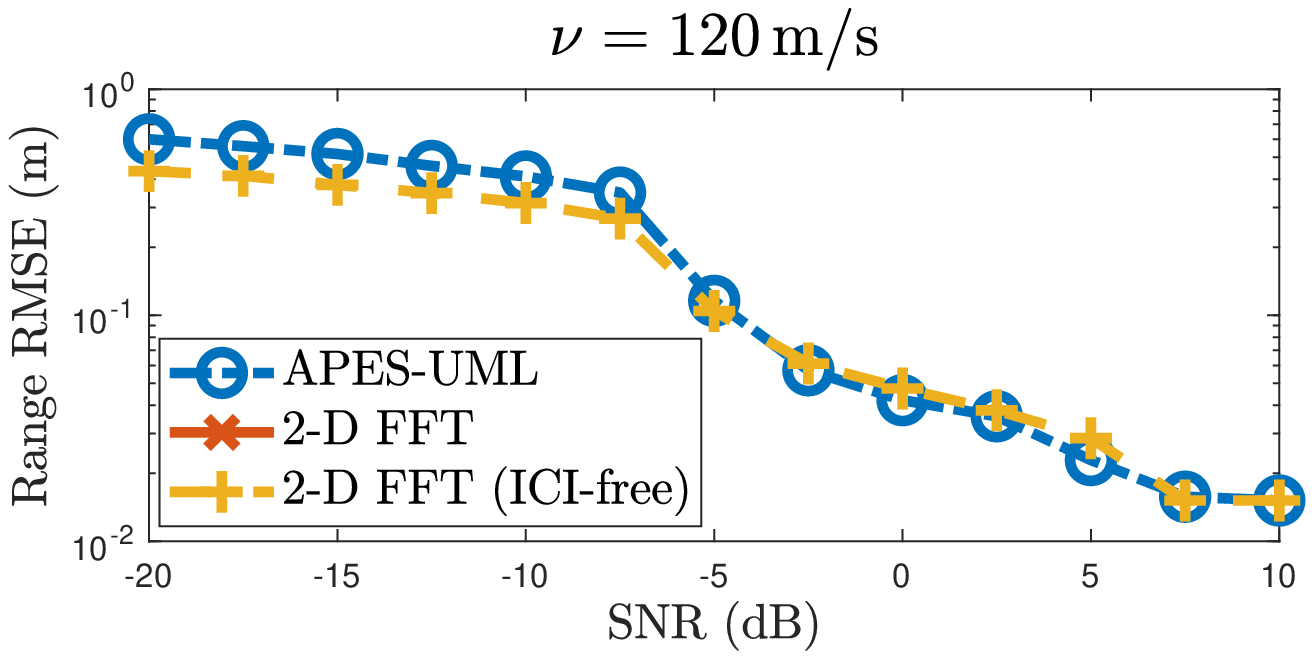}
		}
		
		\end{center}
		\vspace{-0.2in}
        \caption{Range RMSE of the reference target with respect to SNR for \subref{fig_range_rmse_snr_20mps} $\nu = 20 \, \rm{m/s}$, \subref{fig_range_rmse_snr_70mps} $\nu = 70 \, \rm{m/s}$, and \subref{fig_range_rmse_snr_120mps} $\nu = 120 \, \rm{m/s}$ for the scenario in Table~\ref{tab_scenario_mc}.}  
        \label{fig_range_rmse_snr}
        \vspace{-0.15in}
\end{figure}

\begin{figure}
        \begin{center}
        \subfigure[]{
			 \label{fig_vel_rmse_snr_20mps}
			 \includegraphics[width=0.47\textwidth]{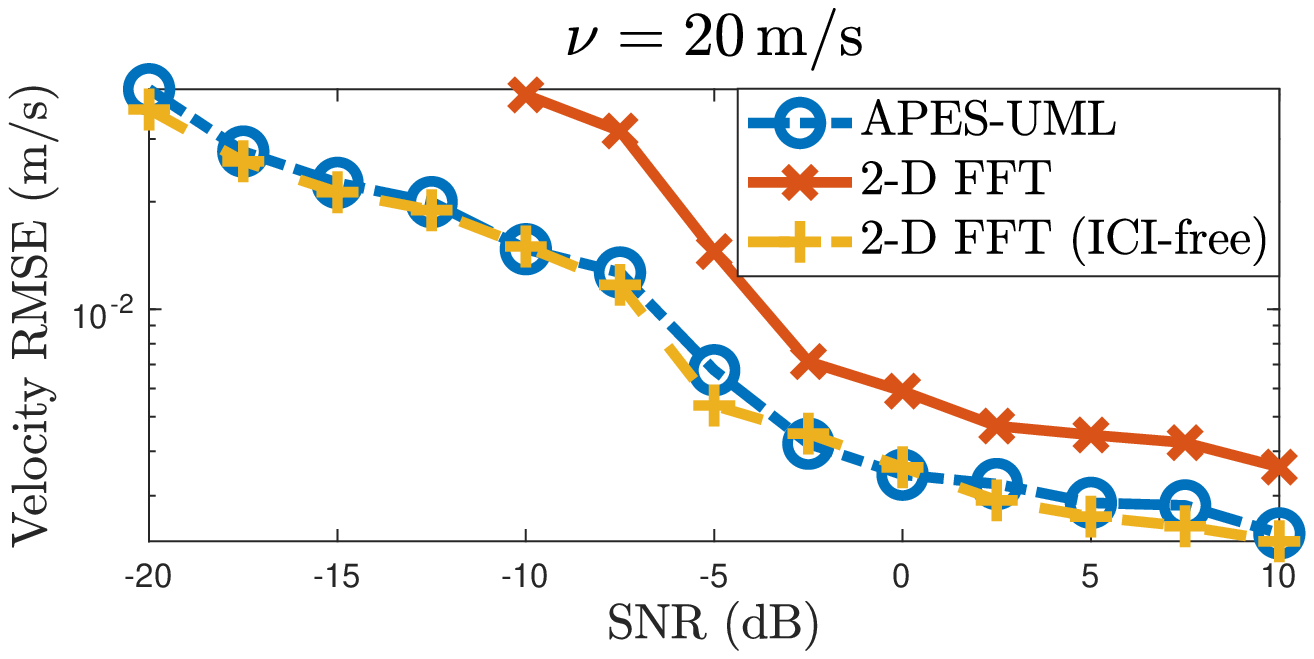}
		}
        \subfigure[]{
			 \label{fig_vel_rmse_snr_70mps}
			 \includegraphics[width=0.47\textwidth]{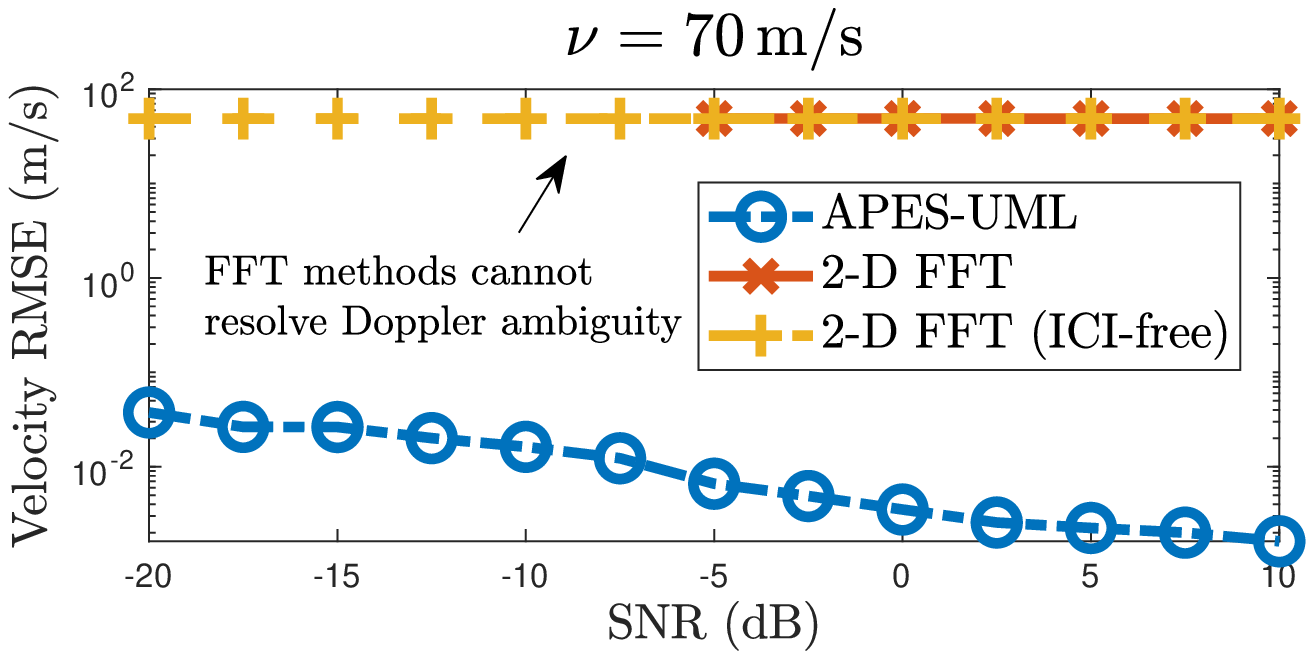}
		}
		
		\subfigure[]{
			 \label{fig_vel_rmse_snr_120mps}
			 \includegraphics[width=0.47\textwidth]{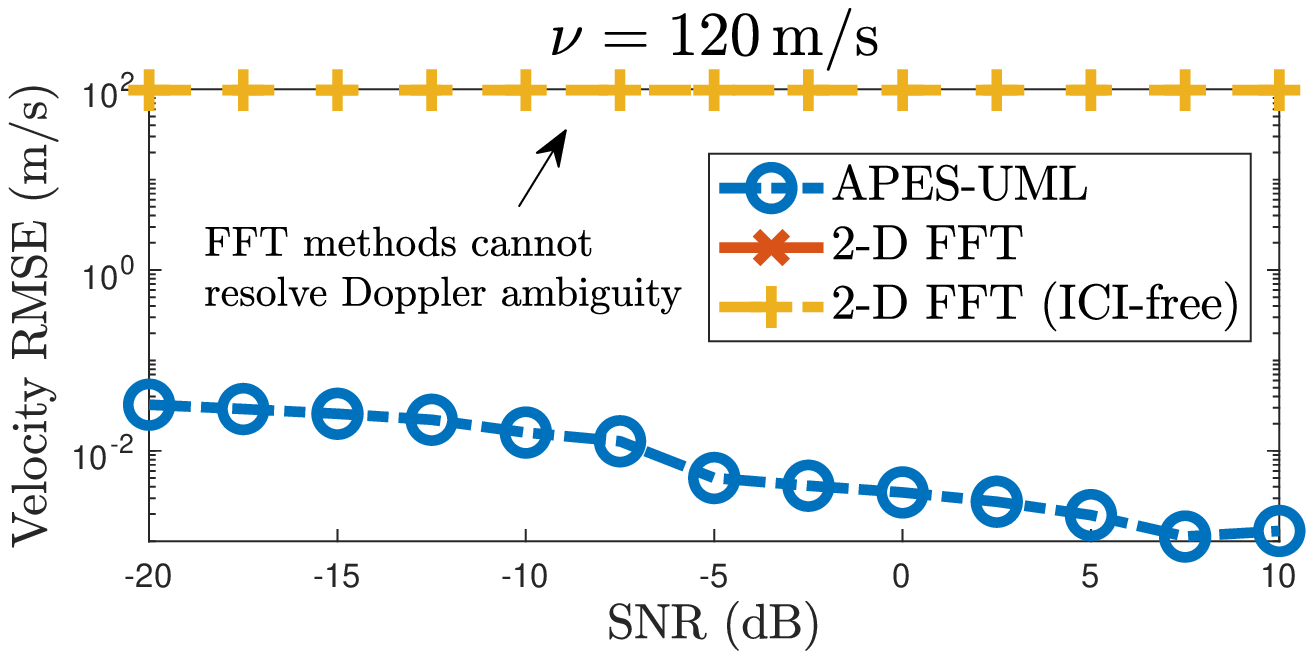}
		}
		
		\end{center}
		\vspace{-0.2in}
        \caption{Velocity RMSE of the reference target with respect to SNR for \subref{fig_vel_rmse_snr_20mps} $\nu = 20 \, \rm{m/s}$, \subref{fig_vel_rmse_snr_70mps} $\nu = 70 \, \rm{m/s}$, and \subref{fig_vel_rmse_snr_120mps} $\nu = 120 \, \rm{m/s}$ for the scenario in Table~\ref{tab_scenario_mc}.}  
        \label{fig_vel_rmse_snr}
        \vspace{-0.12in}
\end{figure}

\vspace{-0.1in}
\section{Concluding Remarks}
We have addressed the multi-target detection/estimation problem for a MIMO-OFDM DFRC system in the presence of non-negligible ICI caused by high-mobility targets. By formulating the ICI-aware sensing as a joint CFO/channel estimation problem, we have developed a novel three-step algorithm for multiple target detection and delay-Doppler-angle estimation. Remarkably, the proposed algorithm can mitigate the ICI effect induced by multiple targets having different radial velocities, which prevents degradation in detection/estimation performance, and at the same time exploit ICI to resolve Doppler ambiguity of the detected targets.
Extensive simulation results have shown that the proposed approach significantly outperforms the traditional FFT based method and can attain the performance achievable in the absence of ICI, with respect to various metrics such as probability of detection and range/velocity estimation accuracy. This indicates that ICI can be successfully suppressed in multi-target scenarios with arbitrary OFDM data symbols, which is crucial for high-speed vehicular JRC applications.
\vspace{-0.1in}




\bibliographystyle{IEEEbib}
\bibliography{ofdm_dfrc}

\newpage 
\renewcommand{\thepage}{}

\section*{\LARGE{Supplementary Material} \scriptsize{For} \normalsize{MIMO-OFDM Joint Radar-Communications: \\Is ICI Friend or Foe?}\\ \footnotesize{Musa Furkan Keskin, Henk Wymeersch, Visa Koivunen }}\label{sec_intro_supp}

\section{Spatial Covariance Matrix of OFDM Radar Observations in \eqref{eq_SCM_obs}}\label{sec_app_scm}

In this part, we derive the SCM of the radar data cube in \eqref{eq_SCM_obs}. Plugging \eqref{eq_ybbar} into \eqref{eq_SCM_obs} yields
\begin{align}\nonumber
\boldR &= \sum_{k_1=0}^{K-1} \sum_{k_2=0}^{K-1} \Bigg[ \alpha_{k_1}^{*} \alpha_{k_2} \, \atx^H(\theta_{k_1})\ftx^{*} \, \atx^T(\theta_{k_2})\ftx \, \arx^{*}(\theta_{k_1})  \, \arx^{T}(\theta_{k_2})  
\\  \nonumber
& ~~\times \bb^H(\tau_{k_1}) \Bigg( \sum_{m=0}^{M-1} \left[ \cc(\nu_{k_1}) \right]_m \left[ \cc^{*}(\nu_{k_2}) \right]_m  \diag{\xx_m}^H \\ \nonumber
 & ~~~~ \times \FF_N \DD^H(\nu_{k_1}) \DD(\nu_{k_2}) \FF_N^H  \diag{\xx_m}    \Bigg) \bb(\tau_{k_2}) \Bigg] \\ \label{eq_scm_derivation1}
    & ~~+ N M \sigma^2 \Imatrix ~,
\end{align}
where we invoke the law of large numbers to approximate the noise covariance as
\begin{align}
    \sum_{m=0}^{M-1} \Zbbar_m^H \Zbbar_m \approx \Eee \left\{ \sum_{m=0}^{M-1} \Zbbar_m^H \Zbbar_m  \right\} = N M \sigma^2 \Imatrix ~,
\end{align}
provided that $M$ and/or $N$ is sufficiently large. Similarly, to ignore signal-noise cross terms in \eqref{eq_scm_derivation1}, the law of large numbers is employed based on the fact that transmit data symbols and noise are uncorrelated.

Using \eqref{eq_ici_D} and the unitary property of the DFT matrix, we can write the direct and cross-target terms separately in \eqref{eq_scm_derivation1}:
\begin{align}\label{eq_all_terms_scm}
\boldR &= \boldR^{\rm{direct}} + \boldR^{\rm{cross}} + N M \sigma^2 \Imatrix ~,
\end{align}
where
\begin{align} \label{eq_r_direct}
\boldR^{\rm{direct}} &= \sum_{k=0}^{K-1} \Bigg[ \abs{\alpha_k}^2 \,  \abs{\atx^T(\theta_k)\ftx}^2 \, \arx^{*}(\theta_k)  \arx^T(\theta_k) \\ \nonumber
 \times  \bb^H(\tau_{k}) & \left( \sum_{m=0}^{M-1} \absbig{ \left[ \cc(\nu_{k}) \right]_m }^2 \diag{\xx_m}^H \diag{\xx_m}  \right)   \bb(\tau_{k}) \Bigg] ~,
\end{align}
\begin{align} \label{eq_r_cross}
& \boldR^{\rm{cross}} 
\\ \nonumber &= \sum_{k_1=0}^{K-1} \sum_{\substack{k_2=0 \\ k_2\neq k_1}}^{K-1} \Bigg[ \alpha_{k_1}^{*} \alpha_{k_2} \, \atx^H(\theta_{k_1})\ftx^{*} \, \atx^T(\theta_{k_2})\ftx \, \arx^{*}(\theta_{k_1})  \, \arx^{T}(\theta_{k_2})  
\\  \nonumber
& ~~\times \bb^H(\tau_{k_1})  \Bigg(  \sum_{m=0}^{M-1} \left[ \cc(\nu_{k_1}) \right]_m \left[ \cc^{*}(\nu_{k_2}) \right]_m \diag{\xx_m}^H \\ \nonumber
 & ~~~~ \times \FF_N \DD(\nu_{k_2}-\nu_{k_1}) \FF_N^H  \diag{\xx_m}   \Bigg) \bb(\tau_{k_2})  \Bigg] ~.
    \end{align}

Using \eqref{eq_boldX_all} and the unit-magnitude property of the steering vector elements in \eqref{eq_steer_delay} and \eqref{eq_steer_doppler}, the direct term in \eqref{eq_r_direct} becomes
\begin{align} \label{eq_r_direct2}
\boldR^{\rm{direct}} &= \norm{\boldX}_F^2 \sum_{k=0}^{K-1} \abs{\alpha_k}^2 \abs{\atx^T(\theta_k)\ftx}^2 \arx^{*}(\theta_k) \arx^T(\theta_k) ~.
\end{align}
Applying the law of large numbers and utilizing the covariance of $\boldX$ in \eqref{eq_cov_x}, we can approximate \eqref{eq_r_direct2} as
\begin{align} \label{eq_r_direct22}
\boldR^{\rm{direct}} &\approx NM \sigmax^2 \sum_{k=0}^{K-1} \abs{\alpha_k}^2 \abs{\atx^T(\theta_k)\ftx}^2 \arx^{*}(\theta_k) \arx^T(\theta_k) ~.
\end{align}

With regard to the cross-term in \eqref{eq_r_cross}, using the properties of the Hadamard product, we can write
\begin{align}\label{eq_A_m}
    \boldA_m &\triangleq \diag{\xx_m}^H \FF_N \DD(\nu_{k_2}-\nu_{k_1}) \FF_N^H  \diag{\xx_m} \\ \nonumber
    &= \left( \xx_m \xx_m^H \right)^{*} \odot  \FF_N \DD(\nu_{k_2}-\nu_{k_1}) \FF_N^H ~.
\end{align}
Assuming sufficiently large $M$, we invoke the law of large numbers to obtain
\begin{align} \label{eq_A_m_exp}
    &\sum_{m=0}^{M-1} \left[ \cc(\nu_{k_1}) \right]_m \left[ \cc^{*}(\nu_{k_2}) \right]_m \boldA_m \\ \nonumber & ~~~~ \approx \sum_{m=0}^{M-1} \left[ \cc(\nu_{k_1}) \right]_m \left[ \cc^{*}(\nu_{k_2}) \right]_m \Eee \{ \boldA_m \} ~,
\end{align}
where expectation is over the distribution of data symbols $\xx_m$. Based on \eqref{eq_A_m} and using the covariance of data symbols in \eqref{eq_cov_x}, we have
\begin{align} \label{eq_A_m_exp2}
    \Eee \{ \boldA_m \} &= \Eee \left\{ \left( \xx_m \xx_m^H \right)^{*}  \right\} \odot  \FF_N \DD(\nu_{k_2}-\nu_{k_1}) \FF_N^H \\ \nonumber
    &= \sigmax^2 \Imatrix \odot  \FF_N \DD(\nu_{k_2}-\nu_{k_1}) \FF_N^H \\ \nonumber
    &= \sigmax^2 \Imatrix \odot  \sum_{n=0}^{N-1} d_n \ff_n \ff_n^H
    \\ \nonumber
    &= \underbrace{ \left(\sigmax^2 + \sum_{n=0}^{N-1} d_n \right) }_{\triangleq \gamma} \Imatrix ~,
\end{align}
where
\begin{align}
    \FF_N &\triangleq \left[ \ff_0 ~ \ldots ~ \ff_{N-1} \right] ~, \\
    d_n &\triangleq \left[ \DD(\nu_{k_2}-\nu_{k_1}) \right]_{n,n} ~.
\end{align}
Inserting \eqref{eq_A_m_exp2} and \eqref{eq_A_m_exp} into \eqref{eq_r_cross} yields
\begin{align} \label{eq_r_cross2}
& \boldR^{\rm{cross}} 
\\ \nonumber &\approx \gamma  \sum_{k_1=0}^{K-1} \sum_{\substack{k_2=0 \\ k_2\neq k_1}}^{K-1} \Bigg[ \alpha_{k_1}^{*} \alpha_{k_2} \, \atx^H(\theta_{k_1})\ftx^{*} \, \atx^T(\theta_{k_2})\ftx \, \arx^{*}(\theta_{k_1})  \, \arx^{T}(\theta_{k_2})  
\\  \nonumber
& ~~~~~~\times \bb^H(\tau_{k_1}) \bb(\tau_{k_2}) \cc^H(\nu_{k_2}) \cc(\nu_{k_1})  \Bigg] ~.
\end{align}
Under the assumption in \eqref{eq_delayDoppler_separation}, the cross-term in \eqref{eq_r_cross2} disappears, i.e.,
\begin{align}\label{eq_cross_zero}
    \boldR^{\rm{cross}} \approx \boldzero_{\Nrx \times \Nrx} ~.
\end{align}
Finally, we insert \eqref{eq_r_direct22} and \eqref{eq_cross_zero} into \eqref{eq_all_terms_scm} to obtain
\begin{align}\nonumber
    \boldR &= NM \sigmax^2 \sum_{k=0}^{K-1} \abs{\alpha_k}^2 \abs{\atx^T(\theta_k)\ftx}^2 \arx^{*}(\theta_k) \arx^T(\theta_k) + N M \sigma^2 \Imatrix ~,
\end{align}
which completes the proof.

\end{document}